\documentclass[11pt, one column]{article}

\usepackage[english]{babel}
\usepackage[utf8x]{inputenc}
\usepackage[T1]{fontenc}

\usepackage[margin=1in]{geometry}

\usepackage{graphicx} 
\usepackage[colorlinks=False]{hyperref} 
\usepackage{amsmath}  
\usepackage{amsfonts} 
\usepackage{amssymb}  
\usepackage{amsthm}
\usepackage[shortlabels]{enumitem}
\usepackage{todonotes}
\usepackage{multirow}
\usepackage{float}
\usepackage{lipsum} % for filler text
\usepackage{pgfplots}
\usepackage{tikz}
\usetikzlibrary{calc}
\usepackage{authblk}
\usepackage{svg}

\usepackage[numbers,sort&compress]{natbib}
\newcommand{\C}{\mathcal{C}}

\newcommand{\F}{\mathcal{F}}

\newcommand{\numl}{k}

\newcommand{\keywords}[1]{\par\noindent\textbf{Keywords: }#1}

\newtheorem{theorem}{Theorem}[section]

\newtheorem{lemma}[theorem]{Lemma}

\newtheorem{definition}{Definition}[section]
\newtheorem{mechanism}{Mechanism}[section]
\usepackage{academicons}
\usepackage{lineno}

\begin{document}

%%
%% The "title" command has an optional parameter,
%% allowing the author to define a "short title" to be used in page headers.
\title{Strategic Facility Location with Limited Liars}

%%
%% The "author" command and its associated commands are used to define
%% the authors and their affiliations.
%% Of note is the shared affiliation of the first two authors, and the
%% "authornote" and "authornotemark" commands
%% used to denote shared contribution to the research.
\author{Yue Gruszecki$^{0009-0002-8037-5998}$}
\author{Elliot Anshelevich$^{0000-0001-9757-6839}$}
\affil{Department of Computer Science, Rensselaer Polytechnic Institute}
\affil{hany4@rpi.edu, eanshel@cs.rpi.edu}

\date{}

%%
%% By default, the full list of authors will be used in the page
%% headers. Often, this list is too long, and will overlap
%% other information printed in the page headers. This command allows
%% the author to define a more concise list
%% of authors' names for this purpose.

%%
%% The abstract is a short summary of the work to be presented in the
%% article.

%% A "teaser" image appears between the author and affiliation
%% information and the body of the document, and typically spans the
%% page.
% \begin{teaserfigure}
%   \includegraphics[width=\textwidth]{sampleteaser}
%   \caption{Seattle Mariners at Spring Training, 2010.}
%   \Description{Enjoying the baseball game from the third-base
%   seats. Ichiro Suzuki preparing to bat.}
%   \label{fig:teaser}
% \end{teaserfigure}

% \received{20 February 2007}
% \received[revised]{12 March 2009}
% \received[accepted]{5 June 2009}

%%
%% This command processes the author and affiliation and title
%% information and builds the first part of the formatted document.
\maketitle

% \linenumbers

\begin{abstract}
We study Nash equilibria in strategic facility location games where clients are located in an arbitrary metric space. Specifically, there are $n$ clients, and the goal is to choose a facility from a set of given locations, so that the total distance from the clients to the facility is as small as possible. While some of the clients are always truthful, $k$ of them are strategic, and will lie about their location if it benefits them. We quantify how the fraction of strategic clients affects the existence and quality of Nash equilibrium and strong equilibrium solutions, and note that even for relatively large $k$, the properties of these solutions can be much better than the results of fully strategyproof mechanisms. 

For Nash equilibrium, we show that it always exists, and the price of stability is very close to 1. More importantly, we prove that all Nash equilibria are within a factor of at most $\frac{n+2k}{n-2k}$ from the optimum solution, and that this price of anarchy bound is almost tight. While strong equilibrium may not exist for this setting, we prove that it always exists for line metrics, and its cost is at most $\frac{n+k}{n-k}$ times that of optimum. 
\end{abstract}

% \ccsdesc[500]{Do Not Use This Code~Generate the Correct Terms for Your Paper}
% \ccsdesc[300]{Do Not Use This Code~Generate the Correct Terms for Your Paper}
% \ccsdesc{Do Not Use This Code~Generate the Correct Terms for Your Paper}
% \ccsdesc[100]{Do Not Use This Code~Generate the Correct Terms for Your Paper}

%%
%% Keywords. The author(s) should pick words that accurately describe
%% the work being presented. Separate the keywords with commas.
\keywords{Facility Location, Nash Equilibrium, Price of Anarchy, Price of Stability}

\section{Introduction}
\label{section:intro}
Strategic facility location problems have been studied under many settings, see Related Work and \cite{chan2021mechanism} for a survey.
%for example \cite{feldman2016voting, alon2010strategyproof,aziz2022strategyproof,balkanski2024randomized,filos2024distortion,lam2024proportional, caragiannis2016truthful, procaccia2013approximate, walsh2021strategy}. 
We consider the case where clients and facilities are located in some arbitrary metric space, and there is a finite set of possible locations to build a facility. The goal is to choose one facility location to minimize the total distance from all clients to the chosen facility (known as the minisum objective). This applies to scenarios such as building a new post office or hospital in a community, or spatial voting settings where the clients are voters and the facilities are candidates, with voters wanting a candidate which is close to them ideologically. Note that if we know the exact locations of the clients, such problems can be solved easily. Strategic facility location, however, considers the case where clients act strategically: the exact locations of the clients are private, while we know only the reported locations of the clients that may not be their true locations. With strategic clients, most existing work studied {\em strategyproof} mechanisms, i.e., mechanisms where no client has an incentive to report anything other than their true location.

%such that no client gets to improve their cost by not reporting their true location (in other words, ). 

Under the strategic facility location or voting settings, most existing work assumes that all clients can lie about their true locations. This is not always true, however.
First, we may know the true locations of many or most of the clients via other means (e.g., their addresses on record), and thus detect when such clients are mis-reporting. Second, even when all clients are able to lie without being detected, it is likely that a large fraction still may not act strategically even if they can benefit from doing so.
%clients would act strategically if they can benefit from it, but this is likely not true in practice. 
For example, as shown in \cite{groseclose2010sincere}, sophisticated voting in Congress is rare. 
%More surprisingly, even if they have the chance to vote strategically, they sometimes still will not. 
In a laboratory setting, 61.8\% of the voters voted for their true preference regardless of the proposed voting rule \cite{lebon2018sincere}.
%(here we note that in this study, voters can act strategically while also being truthful). \elliot{Ask about, make citations consistent} 
In addition, in matching settings, as shown in \cite{pathak2008leveling}, there will be a fixed constant fraction of players who will always play truthfully 
%and a fixed constant fraction of players who will always play sophisticatedly 
in the Boston mechanism.  This means that even for mechanisms that are not strategyproof, the number of strategic players may be limited, and there will be a set of players that will always act truthfully.  Because of this, in our work we assume that only $k$ out of the $n$ clients are strategic, while the other $n-k$ always tell the truth about their locations. This is either because we know their true locations from a different source, because they are inherently truthful people, or because they do not follow any complex and sophisticated strategies, and telling the truth is the simplest and easiest strategy.

Moreover, we note that strategyproof mechanisms for this setting have some limitations, such as that they result in bad approximation bounds. We define the approximation ratio of a mechanism as the worst-case ratio of the outcome produced by the mechanism to the cost of the optimal solution (in other words, this is a way to evaluate the quality of the outcome of the mechanism as compared with the quality of the optimal solution). When the set of possible facility locations is finite,   the approximation ratio for any deterministic strategyproof mechanism is at least $\Omega(n)$, where $n$ is the number of clients (see results in \cite{schummer2002strategy} and discussion in \cite{alon2009strategyproof}). This holds for general metric spaces, and even for circle networks, where the clients are located on a simple cycle.\footnote{Note that results are much better for strategyproof mechanisms if  facilities can be placed anywhere in the space, instead of only at a finite set of locations. For example, as shown in \cite{gravin2025approximation}, there exist 3-approximation deterministic strategyproof mechanisms when clients are located in $\mathbb{R}^d$ and the distance corresponds to most $L_p$ norms. 
%an $d$-dimensional space with Euclidean norm $L_p, p \in \mathbb{R}_{\geq 1} \cup\{\infty\}$ as a distance function, for any $d > 1$. 
In addition, on a tree network, one can obtain the optimal solution using a deterministic strategyproof mechanism \cite{schummer2002strategy}.}
This bound is significantly worse than the best possible approximation ratio of 3 that strategyproof mechanisms can achieve in one-dimensional space \cite{feldman2016voting}. 
% \yue{This is significantly worse than the best possible approximation ratio of 3 that strategyproof mechanisms can achieve in an $d$-dimensional space with Euclidean norm $L_p, p \in \mathbb{R}_{\geq 1} \cup\{\infty\}$ as a distance function, for any $d > 1$, as shown in \cite{feldman2016voting, walsh2021strategy, gravin2025approximation}.}
% one-dimensional space \cite{feldman2016voting, walsh2021strategy}. 
Thus, despite their obvious benefits, strategyproof mechanisms may not be desirable for metric spaces beyond 1-D.
%This means that in general metric space and with limited facility locations, the approximation ratio for any deterministic strategyproof mechanism would be unbounded. 
In our work, we instead relax the constraint of strategyproofness and study Nash equilibrium and strong equilibrium for non-strategyproof mechanisms. In addition, as discussed above, using the fact that not every client would act strategically, we analyze the stable solutions with the assumption that only $k$ of the clients act strategically. This can also be seen as another relaxation of stategyproofness. As a result, we show that Nash equilibrium always exists and the quality of the equilibria is close to the optimal solution, even for a relatively large fraction (e.g., 20-40\%) of strategic clients, whereas the approximation ratio is $\Omega(n)$ for strategyproof mechanisms. We present the detailed results below.

\begin{figure}[!ht]
    \centering
    %\includesvg[width = 4in]{results.svg}
    \includegraphics[width=3in]{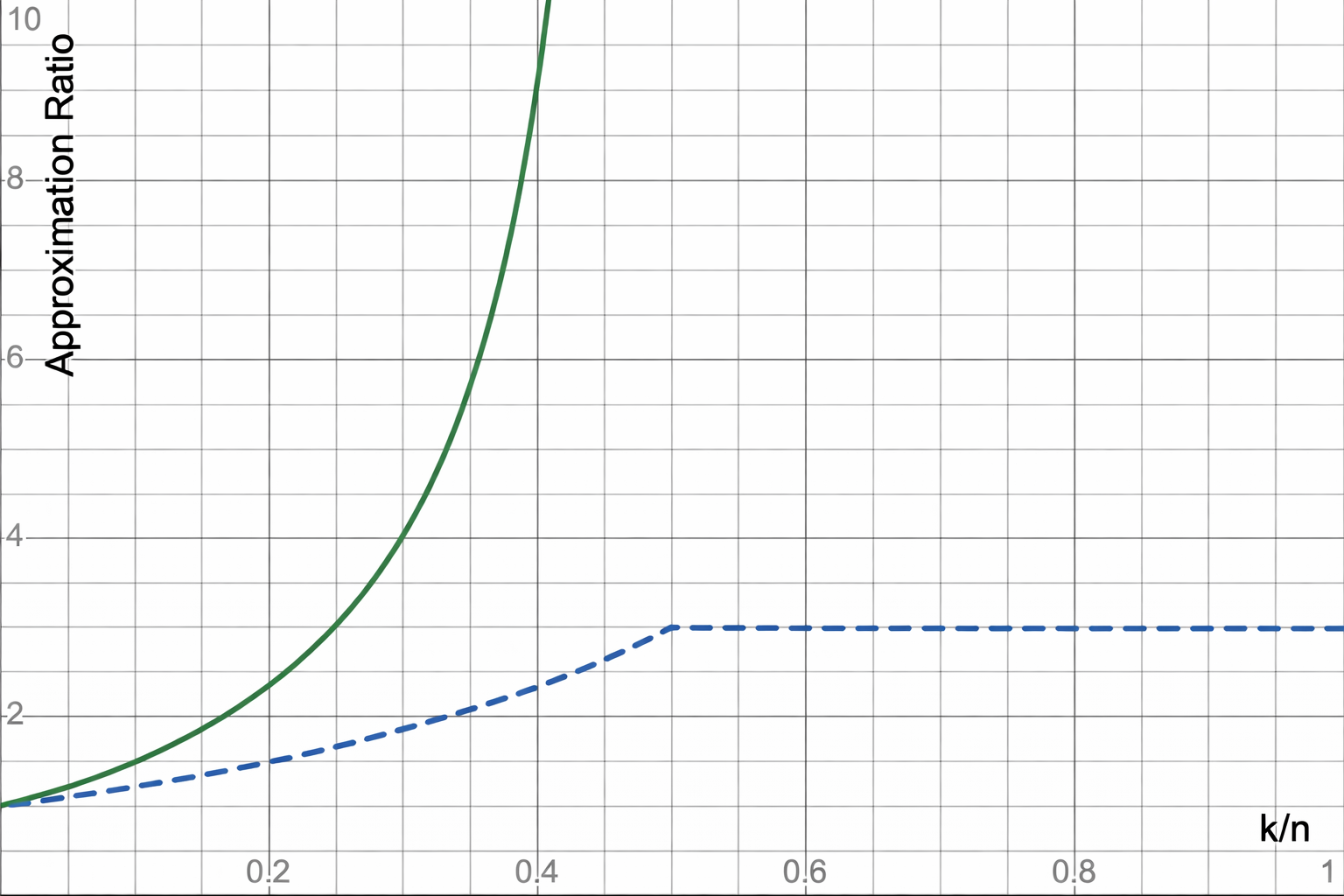}
    \caption{Plots of the approximation ratio w.r.t. $k/n$ under the worst case NE (for general metrics) or SNE (for line metrics). The solid line shows the upper bound for POA w.r.t. $k/n$ while the dashed line shows the upper bound for SPOA w.r.t. $k/n$.}
    \label{fig:results}
\end{figure}

\subsection{Our Contributions}

We analyze the natural mechanism that, given the clients' reported locations, simply chooses the solution minimizing the (perceived) sum of distances from the clients to the chosen facility. Despite its simplicity, we show that this mechanism has nice properties, and in fact no other reasonable deterministic mechanism can achieve better price of anarchy bounds.

We first show that there always exists a (pure) Nash equilibrium (NE) for a general metric space with an arbitrary finite set of given facility locations. We then analyze the quality of such equilibria: the approximation ratio of the Nash equilibrium with smallest cost (known as the price of stability (POS)) is upper bounded by $\frac{n+1}{n-1}$ when $k = 1$ and by $\frac{n+k/(k-1)}{n-k/(k-1)}$ when $k \geq 2$. Note that both of these bounds are upper bounded by $\frac{n+2}{n-2}$, meaning that when $n$ is large, this is {\em extremely} close to 1, and the best NE is almost as good as the optimal solution. More importantly, for the approximation ratio of the Nash equilibrium with the largest cost (known as the price of anarchy (POA)), we prove that it is upper bounded by $\frac{n+2k}{n-2k}$ for $ k < \frac{n}{2}$. %but is unbounded for $k \geq n/2$,
We show that this bound is almost tight, and provide a lower bound of $\frac{n+2k-2}{n-2k+2}$ for the POA of {\em any} reasonable\footnote{By ``reasonable'', we mean anonymous and monotonic, see Theorem \ref{thm:poatight}.} mechanism. The relationship between the fraction of strategic clients w.r.t. the total number of clients and the approximation ratio is shown in Figure \ref{fig:results}. Thus, while strategyproof mechanisms give an $\Omega(n)$ approximation ratio for general metric spaces, relaxing strategyproofness results in all Nash equilibria having an approximation ratio of a small constant, as long as the fraction of strategic clients is not too large (e.g., less than 40 percent). %Even for line metrics, the quality of Nash equilibrium improves on the performance of strategyproof mechanisms when $k$ is less than 25 percent of the total, since the price of anarchy becomes less than 3 at that point.

%Compared to strategyproof mechanisms, which have an unbounded approximation ratio in a general metric space, when the fraction of strategic clients is less than 50\% of the total clients, even the worst NE is bounded, for example, if 40\% of the clients act strategically, the approximate ratio using the worst NE is 9. And when 25\% of the clients act strategically, the approximate ratio using the worst NE is only 3. 

So far, we only considered Nash equilibria, where only one client can alter their reported location at a time. What if a group of clients can collude and alter their reported locations at the same time, however, so that all of them can improve their costs? To address this, we next study strong Nash equilibrium (SNE) \cite{andelman2009strong,epstein2007strong}. Unfortunately, a SNE does not always exist for general metrics; because of this we restrict our attention to one-dimensional line metrics. We prove that SNE always exists for such metrics. Furthermore, we show that for a line metric, the upper bound for the approximation ratio using the best SNE and the worst SNE (known as the strong price of stability (SPOS) and the strong price of anarchy (SPOA), respectively) are both $\frac{n+k}{n-k}$ for $k \leq \frac{n}{2}$ and 3 for $k > \frac{n}{2}$. Note that this is a significantly better result than POA on a line. Compare this also to the best achievable bound of 3 for any strategyproof mechanism for a line metric \cite{feldman2016voting, walsh2021strategy}. We show that any SNE would give an approximation bound strictly smaller than 3 when less than half of the clients act strategically, and when $k$ is only 25 percent, the approximation ratio becomes only $\frac{5}{3}$ (see Figure \ref{fig:results}).

%The relationship between the fraction of strategic clients w.r.t. the total number of clients and the approximation ratio is shown in Figure \ref{fig:results}.

\subsection{Related Work}
Strategic facility location problems have been studied for decades, see for example, \cite{filos2024distortion, tang2023strategyproof, balkanski2024randomized, feldman2016voting, alon2010strategyproof,aziz2022strategyproof,lam2024proportional, caragiannis2016truthful, procaccia2013approximate, walsh2021strategy, aloupis2010improved}, and see  \cite{chan2021mechanism} for a survey. Most of such work focused on strategyproof mechanisms for choosing one facility. On a line, when the facility can be placed anywhere, choosing the median client is a strategyproof mechanism; it is optimal for the minisum objective, and is a 2-approximation for minimax, which is also the best achievable among deterministic strategyproof mechanisms \cite{procaccia2013approximate}. Beyond a line, as shown in \cite{gravin2025approximation}, there exist 3-approximation deterministic strategyproof mechanisms for clients located in $\mathbb{R}^d$ and most $L_p$ metrics. In addition, on a tree network, one can obtain the optimal solution using a deterministic strategyproof mechanism \cite{schummer2002strategy}.
We consider a more general version of strategic facility location problems, where facilities can only be chosen from a set of given locations instead of anywhere in the space. Note that while we define our model as having a finite set of possible facility locations, all of our upper bounds (with one exception) apply to the case when the set of possible facility locations is infinite as well, and thus include the case when facilities can be placed anywhere as a special case. (The exception to this is the POS bound in Theorem \ref{thm:gen_pos}, which only holds when the set of facility locations is finite.) For this setting where not all facility locations are possible, \cite{feldman2016voting, walsh2021strategy} showed that no deterministic strategyproof mechanism can achieve an approximation ratio better than 3 for either minimax or minisum objectives, even for clients located on a line. 

For metric spaces beyond a line metric, it is shown in \cite{ schummer2002strategy} that the approximation ratio for any deterministic strategyproof mechanism is at least $\Omega(n)$ for minisum in circle networks, where $n$ is the number of clients. %This shows that the approximation ratio can be unbounded for general metric spaces for any deterministic strategyproof mechanism. 
However, using randomized algorithms can improve the result in general metric spaces with limited facility locations. As is shown in \cite{anshelevich2017randomized}, random dictatorship can achieve an approximation ratio of $3 - \frac{2}{n}$ for minisum in a general metric space, which improves to $2-\frac{2}{n}$ if facilities can be placed on top of any client \cite{alon2009strategyproof}. Due to the incompatibility of strategyproofness and other properties such as proportional fairness, \cite{lam2024proportional, aziz2022strategyproof} also analyze the properties and qualities of Nash equilibria under different mechanisms. Much like their work, in this paper, instead of being restricted to deterministic strategyproof mechanisms, we study the properties and stability of a deterministic mechanism w.r.t. Nash equilibria and strong Nash equilibria. As a result, we are able to obtain small approximation ratios when the number of strategic clients is limited.

To our knowledge, no previous work has been done studying facility location with a limited number of strategic players. Such considerations have been previously studied in a matching setting \cite{pathak2008leveling}, which looks at equilibria with sets of truthful players and sophisticated players in the Boston mechanism. 

% Even though under the strategic facility location or voting setting, most studies assume that all clients would act strategically if they can benefit from it, it is not true in practice. For example, \cite{groseclose2010sincere} shows that sophisticated voting in Congress is rare. More surprisingly, even if they have the chance to vote strategically, they sometimes still will not. \cite{lebon2018sincere} shows that in a laboratory setting, the single-vote system leads to more voters acting sincerely than multi-vote rules. In fact, 61.8\% of the voters are sincere regardless of the proposed voting rule (here we note that in this study, voters can act strategically while also being sincere). In addition, \cite{pathak2008leveling} shows that there will be a fixed constant of players who will always play sincerely and a fixed constant of players who will always play sophisticatedly in the Boston mechanism. Therefore, in this work, we consider the case where not all clients/voters would act strategically, but only a fraction of them, while the others would always act sincerely. 

\section{Preliminaries and Notation}
\label{section:prelim}
Consider facility location problems where we are given a set of $n$ client locations $\C$, and a finite\footnote{Note that all of our upper bounds hold for infinite sets as well, except for the POS bound in Theorem \ref{thm:gen_pos} which uses the finite assumption.} set of possible facility locations $\F$. All of these are located in a metric space with distance $d$, with multiple clients allowed to be located at the same point. For tie-breaking purposes, we assume that there is some fixed
ordering for the facility locations $\F$; when a location comes before another in this ordering, we will say that it is {\em lower-numbered}.  (If all clients and facilities are located on a line, we assume that the facility locations are simply ordered from left to right, with the leftmost facility being the lowest-numbered and the rightmost facility being the highest-numbered.)

The goal is to place a facility that would minimize the total distance from each client to the facility. Formally, for choosing location $A$, we define the cost of each client $i$ as the distance between $i$ and $A$, $d(i,A)$, and the {\em optimal solution} is a facility location $A$ that minimizes $\sum_{i\in \C}d(i,A).$

The clients can act strategically. However, only $\numl$ of them can misreport their location (by instead reporting any location in the metric space), while the others will always report truthfully. We denote the set of clients that act strategically by $L$ and the set of clients that always act truthfully by $T$. Note that we have $|T| = n - \numl, |L| = \numl, T \cap L = \emptyset, T \cup L = \C$. With some abuse of notation, for client $i \in \C$, we will denote $i$'s reported location by $x_i$ and their true location by $i$. Let $\vec{t}$ be the vector of locations of clients in $T$, and $\vec{x}$ be the vector of the {\em reported} locations of clients in $L$. Then, in order to choose a facility, the information available to our mechanism is exactly $(\vec{t}; \vec{x})$, although the mechanism does not know which clients are truthful and which are strategic. Based on this available information, we define the {\em perceived cost} of a facility $A$ with respect to reported locations $(\vec{t}; \vec{x})$ to be 

$$cost_A(\vec{t};\vec{x}) = \sum_{i \in T}d(i,A) + \sum_{i \in L}d(x_i,A).$$

The mechanism which we use in our paper is extremely simple: it just picks the location with smallest perceived cost (in case of ties it chooses the lower-numbered location in $\F$). Note that while this mechanism is simple and intuitive, we show that it has the best possible price of anarchy guarantees, as compared with any other reasonable mechanism. Because of this, we believe that there is no reason to form a more complex mechanism, when this one already performs the best with respect to equilibrium quality, and has the advantage of being natural and intuitive. Finally, we remark that while our mechanism is defined for general metric spaces, even for a line metric it is {\em not} the same as the well-known median mechanism, which chooses the facility closest to the median agent. In fact, the median mechanism on a line can form solutions which are a factor of 3 away from optimum, while we show our mechanism significantly improves on this factor even for relatively large $k$. We call this mechanism $GM^*$, since when facilities can be chosen anywhere in the space, this mechanism would choose the geometric median of the reported client locations. We formally define $GM^*$ as follows. 

\begin{mechanism}
	Given the reported client locations, $(\vec{t}; \vec{x})$, our mechanism $GM^*$ returns the facility $A$ minimizing $cost_A(\vec{t};\vec{x})$.
    If there is a tie, choose the facility with the lowest number w.r.t. to the facility ordering.
\label{mech:opt}
\end{mechanism}

We use $GM^*(\vec{t}; \vec{x})$ to denote the outcome of $GM^*$ given reported locations $(\vec{t}; \vec{x})$. We will also sometimes use notation 
$$cost_{GM^*}(\vec{t};\vec{x}) = \sum_{i \in T}d(i,C) + \sum_{i \in L}d(x_i,C),$$
where $C = GM^*(\vec{t};\vec{x})$ to refer to the cost of the outcome of $GM^*$.

%For simplicity, we denote the input of $GM^*$ by $(\vec{t}; \vec{x})$ where $\vec{t}$ is the vector of locations of clients in $T$ and $\vec{x}$  is the vector of the reported locations of clients in $L$. In addition, with some abuse of notation, for client $i \in \C$, we will denote $i$'s reported location by $x_i$ and their true location by $i$. 
When dealing with equilibria and some clients changing their reported locations, we will often use the following standard notation.
We denote the $(k-1)-$dimensional vector of reported locations of all clients in $L\setminus\{i\}$, by $\vec{x}_{-i}$. For example, the perceived cost of $A$ when all clients in $L$ report location $x_i$, but client $i$ reports location $x_i'$, is
$$cost_A(\vec{t};\vec{x}_{-i},x_i') = \sum_{j \in T}d(j,A) + \sum_{j \in L\setminus\{i\}}d(x_j,A) + d(x_i',A).$$
Similarly, for any set of clients $I$, $\vec{x}_{-I}$ refers to the vector of $x_i$ values for clients in $L\setminus I$. 

We are interested in quantifying the quality of equilibrium solutions for our setting.
Recall the definitions of (pure) Nash equilibrium and strong Nash equilibrium, as well as the notions to evaluate their quality:

\begin{definition}
	Let $S_i$ be the set of all possible locations client $i$ can report and $M$ be a mechanism. We say that $(\vec{t}; \vec{x})$ is a Nash equilibrium (NE) if for each client $i \in L$, for all $x'_i \in S_i$, we have
	$$d(i,M(\vec{t}; \vec{x}_{-i},x'_i)) \geq d(i,M(\vec{t}; \vec{x})).$$
\end{definition}  

\begin{definition}
	Let $S_I$ be the set of sets of all possible locations that a set of clients $I$ can report and $M$ be a mechanism. We say that $(\vec{t}; \vec{x})$ is a Strong Nash equilibrium (SNE) if for each set of clients $I \subseteq L$, for all $x'_I \in S_I$, there exists a client $i \in I$ such that
	$$d(i,M(\vec{t}; \vec{x}_{-I},x'_I)) \geq d(i,M(\vec{t}; \vec{x})).$$
\end{definition}  

\begin{definition}
    Let $S$ be the set of all facilities $A$ such that there exists a NE under mechanism $M$ which selects $A$. Let $O$ be the optimal solution. We define the price of stability (POS) as
    $$\inf_{A\in S} \frac{\sum_{i \in \C}d(i, A)}{\sum_{i \in \C}d(i, O)},$$
    and the price of anarchy (POA) as
    $$\sup_{A\in S} \frac{\sum_{i \in \C}d(i, A)}{\sum_{i \in \C}d(i, O)}.$$
\end{definition}

\begin{definition}
    Let $S$ be the set of all facilities $A$ such that there exists a SNE under mechanism $M$ which selects $A$. Let $O$ be the optimal solution. We define the strong price of stability (SPOS) as
    $$\inf_{A\in S} \frac{\sum_{i \in \C}d(i, A)}{\sum_{i \in \C}d(i, O)},$$
    and the strong price of anarchy (SPOA) as
    $$\sup_{A\in S} \frac{\sum_{i \in \C}d(i, A)}{\sum_{i \in \C}d(i, O)}.$$
\end{definition}

The above definitions provide a measure of Nash equilibrium and strong Nash equilibrium quality for a particular instance. As standard in the literature, we will also refer to the (strong) price of stability or price of anarchy of a {\em mechanism} as the worst-case PoS or PoA over all possible input instances.  

% In addition, we define the cost of choosing a facility as follows:

% \begin{definition}
%     Let
%     $$cost_A(\vec{t};\vec{x}) = \sum_{i \in T}d(i,A) + \sum_{i \in L}d(x_i,A),$$
%     $$cost_A(\vec{t};\vec{x}_{-i},B) = \sum_{j \in T}d(j,A) + \sum_{j \in L\setminus\{i\}}d(x_j,A) + d(B,A),$$
%     $$cost_{GM^*}(\vec{t};\vec{x}) = \sum_{i \in T}d(i,C) + \sum_{i \in L}d(x_i,C),$$
%     where $C = GM^*(\vec{t};\vec{x})$.
% \end{definition}

\section{Nash Equilibrium}
\label{sec:ne}
We first consider Nash equilibrium in a general metric space with $\numl$ liars. We begin by observing the following useful lemma.

\begin{lemma}
	For any $\vec{x}$ with $GM^*(\vec{t};\vec{x}) = B$, and for any client $i\in L$, we have that 
    $cost_{GM^*}(\vec{t};\vec{x}_{-i},B)$ $ = cost_{B}(\vec{t};\vec{x}_{-i},B)$.
    In other words, if $B$ is the winner with reported locations $\vec{x}$, then client $i$ can report being directly at $B$, so that either $B$ is still the winner, or ties with the new winner. This also means that for all $C\in\F$, $cost_{B}(\vec{t};\vec{x}_{-i},B) \leq cost_{C}(\vec{t};\vec{x}_{-i},B)$.	\label{lemma:lie-on-top}
\end{lemma}

\begin{proof}
	Since $GM^*(\vec{t};\vec{x}) = B$, we have that for any $C \in \F$, 
	\[
	\begin{aligned}
		\sum_{j\in T} d(x_j,B) + \sum_{j \in L \setminus \{i\}} d(x_j, B) + d(x_i, B) &\leq \sum_{j\in T} d(x_j,C)+ \sum_{j \in L \setminus \{i\}} d(x_j, C) \\
        &+ d(x_i, C)
	\end{aligned}
	\]
	Rearrange the above inequality, and we have 
	\[
	\begin{aligned}
		\sum_{j\in T} d(j,B) + \sum_{j \in L \setminus \{i\}} d(x_j, B)  &\leq \sum_{j\in T} d(j,C) + \sum_{j \in L \setminus \{i\}} d(x_j, C) + d(x_i, C) - d(x_i, B)\\
		&\leq \sum_{j\in T} d(j,C) + \sum_{j \in L \setminus \{i\}} d(x_j, C) + d(x_i, B)+ d(B,C) - d(x_i, B)\\
		&= \sum_{j\in T} d(j,C) + \sum_{j \in L \setminus \{i\}} d(x_j, C) + d(B,C)
	\end{aligned}
	\]
	The last two inequalities hold because all clients and facilities are in some metric space so we can utilize triangle inequality. Note that the above inequality means 
	$$\sum_{j\in T} d(j,B) +\sum_{j \in L \setminus \{i\}} d(x_j, B) + d(B,B)  \leq \sum_{j\in T} d(j,C) + \sum_{j \in L \setminus \{i\}} d(x_j, C) + d(B,C)$$
	for all $C \in \F$, and thus $cost_B(\vec{t};\vec{x}_{-i},B)\leq cost_{C}(\vec{t};\vec{x}_{-i},B),$ as desired. We can then conclude that either $GM^*(\vec{t};\vec{x}_{-i},B) = B$ or $GM^*(\vec{t};\vec{x}_{-i},B) = B'$, with $cost_{B'}(\vec{t};\vec{x}_{-i},B) = cost_{B}(\vec{t};\vec{x}_{-i},B)$. 
\end{proof}

% \begin{proof}
% 	Since $GM^*(\vec{t};\vec{x}) = B$, we have that for any $C \in \F$, 
% 	\[
% 	\begin{aligned}
% 		\sum_{j\in T} d(x_j,B) + \sum_{j \in L \setminus \{i\}} d(x_j, B) + d(x_i, B) \leq \sum_{j\in T} d(x_j,C)+ \sum_{j \in L \setminus \{i\}} d(x_j, C) + d(x_i, C)
% 	\end{aligned}
% 	\]
% 	Rearrange the above inequality, and we have 
% 	\[
% 	\begin{aligned}
% 		\sum_{j\in T} d(j,B) + \sum_{j \in L \setminus \{i\}} d(x_j, B)  &\leq \sum_{j\in T} d(j,C) + \sum_{j \in L \setminus \{i\}} d(x_j, C) + d(x_i, C) - d(x_i, B)\\
% 		&\leq \sum_{j\in T} d(j,C) + \sum_{j \in L \setminus \{i\}} d(x_j, C) + d(x_i, B) + d(B,C)- d(x_i, B)\\
% 		&= \sum_{j\in T} d(j,C) + \sum_{j \in L \setminus \{i\}} d(x_j, C) + d(B,C)
% 	\end{aligned}
% 	\]
% 	The last two inequalities hold because all clients and facilities are in some metric space so we can utilize triangle inequality. Note that the above inequality means 
% 	$$\sum_{j\in T} d(j,B) +\sum_{j \in L \setminus \{i\}} d(x_j, B) + d(B,B)  \leq \sum_{j\in T} d(j,C) + \sum_{j \in L \setminus \{i\}} d(x_j, C) + d(B,C)$$
% 	for all $C \in \F$, and thus $cost_B(\vec{t};\vec{x}_{-i},B)\leq cost_{C}(\vec{t};\vec{x}_{-i},B),$ as desired. We can then conclude that either $GM^*(\vec{t};\vec{x}_{-i},B) = B$ or $GM^*(\vec{t};\vec{x}_{-i},B) = B'$, with $cost_{B'}(\vec{t};\vec{x}_{-i},B) = cost_{B}(\vec{t};\vec{x}_{-i},B)$. 
% \end{proof}

Using the above lemma, we are able to show that if $O^* = GM^*(\vec{t})$ (i.e., $GM^*$ applied only to the locations of $T$), then all strategic agents in $L$ reporting their location to be at $O^*$ results in a Nash equilibrium. This proves the following theorem.

\begin{theorem}
	Using Mechanism $GM^*$, there always exists a Nash equilibrium under any metric space with distance $d$.
\label{thm:neexist}
\end{theorem}

% \begin{lemma}
% 	When , there exists a Nash equilibrium under $GM^*$.
% 	\label{lemma:ne1} 
% \end{lemma}
\begin{proof}
First, consider the case when $k = 1$. Let $L = \{i\}$, and let $x_i$ be the location which minimizes $d(GM^*(\vec{t};x_i),i)$, i.e., the location which, if reported by $i$, results in the smallest cost for $i$ under the mechanism $GM^*$. Then $x_i$ is trivially a Nash equilibrium for the case when  $k = 1$.

Now suppose that $k>1$. Let $O^* = GM^*(\vec{t})$ (i.e., $GM^*$ applied only to the locations of $T$) and $\vec{x} = (O^*, O^*, \cdots, O^*)$. We will prove that $cost_{GM^*}(\vec{t};\vec{x}) = cost_{O^*}(\vec{t};\vec{x})$, and thus $\vec{x}$ is a Nash equilibrium for $k > 1$ under $GM^*$. 
	
We note that since $O^* = GM^*(\vec{t})$, we have that 
	$$O^* = \text{argmin}_{B \in \F}\sum_{i \in T}d(i,B).$$
	With all the liars reporting the same location $O^*$, $\vec{x} = (O^*, O^*, \cdots, O^*)$, it is obvious that 
	$$O^* = \text{argmin}_{B \in \F}cost_B(\vec{t};\vec{x}).$$
	 Let $O'= GM^*(\vec{t};\vec{x})$, we will show that $O' = O^*$. Assume otherwise, $O'\neq O^*$. We observe that 
     $$cost_{O^*}(\vec{t};\vec{x}) = \sum_{j \in T} d(j, O^*) + \sum_{j \in L}d(O^*,O^*),$$
     $$cost_{O'}(\vec{t};\vec{x}) = \sum_{j \in T} d(j, O') + \sum_{j \in L}d(O',O^*).$$
     Since $O'\neq O^*$, we have that $\sum_{j \in L}d(O',O^*) > 0$. Since $O^* = GM^*(\vec{t})$, we also know that $\sum_{j \in T} d(j, O^*)\leq \sum_{j \in T} d(j, O').$ Thus, we have that $cost_{O'}(\vec{t};\vec{x}) > cost_{O^*}(\vec{t};\vec{x})$, meaning that  $GM^*(\vec{t};\vec{x}) \neq O'$ as the cost for $O^*$ is strictly smaller than $O'$. This is a contradiction. Therefore, we must have $O' = O^*, GM^*(\vec{t};\vec{x}) = O^*$.
      
    We will now show that no client can alter their reported location from $\vec{x}$ so that another location would win. Suppose to the contrary, that there exists client $i \in L$ such that $GM^*(\vec{t};\vec{x}_{-i}, x'_i) = A\neq O^*$. By Lemma \ref{lemma:lie-on-top}, we can assume w.l.o.g. that $GM^*(\vec{t};\vec{x}_{-i}, A) = A'$, with $cost_A(\vec{t};\vec{x}_{-i}, A) = cost_{A'}(\vec{t};\vec{x}_{-i}, A)$.
     % \sum_{j\in T}d(j, a') + \sum_{j\in T\setminus \{i\}}d(x_j, a') + d(a, a') = \sum_{j\in T}d(j, a) + \sum_{j\in T\setminus \{i\}}d(x_j, a)$.
     However, since we know that $O^* = \text{argmin}_{B \in \F}\sum_{i \in T}d(i,B), $ we have that
	 % By Lemma \ref{lemma:lie-on-top}, we can assume w.l.o.g. we can assume for contradiction that for any $i$, they can alter the optimal solution to be $a_i$ by reporting their location at $a_i$. However, since we know that $o^*$ is the optimal solution for $T$, we have that
	\[
	\begin{aligned}
        cost_{O^*}(\vec{t};\vec{x}_{-i},A) &=
		\sum_{j \in T} d(j, O^*) + \sum_{j \in L\setminus\{i\}}d(O^*,O^*) + d(A, O^*)\\ 		
		&\leq \sum_{j \in T} d(j, A)  + d(A, O^*)\\
		&\leq \sum_{j \in T} d(j, A) + \sum_{j \in L\setminus\{i\}}d(A,O^*)\\
		&= cost_A(\vec{t};\vec{x}_{-i},A)\\
        &= cost_{A'}(\vec{t};\vec{x}_{-i},A)
	\end{aligned}
	\]
	Note that when $k > 2$, the last inequality is strictly less than, which means $GM^*(\vec{t};\vec{x}_{-i},A) \neq A'$, since $O^*$ has a strictly better perceived cost, giving us a contradiction. We then consider the case when $k = 2$. If $\sum_{j \in T} d(j, O^*) < \sum_{j \in T} d(j, A)$, then the first inequality above is strict and thus $GM^*(\vec{t};\vec{x}_{-i},A) \neq A'$, so assume otherwise, $\sum_{j \in T} d(j, O^*) = \sum_{j \in T} d(j, A)$. We first note that since $O^* = GM^*(\vec{t}), \sum_{j \in T} d(j, O^*) = \sum_{j \in T} d(j, A)$, this must mean that $O^*$ is lower numbered than $A$. %We will then show that $A = A'$. Assume otherwise, $A \neq A'$. 
    Now consider the cost of $A'$.
    For the case when $k=2$,
    $$cost_A(\vec{t};\vec{x}_{-i},A) = \sum_{j \in T} d(j, A)  + d(A, O^*),$$
    $$cost_{A'}(\vec{t};\vec{x}_{-i},A) = \sum_{j \in T} d(j, A')  + d(A', O^*) + d(A, A'),$$
    Since $O^* = \text{argmin}_{B \in \F}\sum_{i \in T}d(i,B),$ $ \sum_{j \in T} d(j, O^*) = \sum_{j \in T} d(j, A)$, we have that $\sum_{j \in T} d(j, A) 
    \leq\sum_{j \in T} d(j, A')$. By triangle inequality, we also have that $d(A,O^*)\leq d(A,A')+d(A',O^*)$. But by our definition of $A'$, we also know that $cost_A(\vec{t};\vec{x}_{-i}, A) = cost_{A'}(\vec{t};\vec{x}_{-i}, A)$. Thus, it must be that $\sum_{j \in T} d(j, O^*) = \sum_{j \in T} d(j, A) = \sum_{j \in T} d(j, A')$. Therefore, since $A'\neq GM^*(\vec{t})$, it must be that $O^*$ is lower numbered than $A'$. Finally, for this case when $k=2$, we also have that $cost_{O^*}(\vec{t};\vec{x}_{-i},A) = cost_A(\vec{t};\vec{x}_{-i},A) = cost_{A'}(\vec{t};\vec{x}_{-i},A)$, so  $O^*$ would still be picked by our mechanism instead of $A'$. 
    
    % %Since we have that $A \neq A'$, we must have $d(A, A') > 0$. In addition, since $O^* = \text{argmin}_{B \in \F}\sum_{i \in T}d(i,B),$ $ \sum_{j \in T} d(j, O^*) = \sum_{j \in T} d(j, A)$, we have that $\sum_{j \in T} d(j, A) 
    % \leq\sum_{j \in T} d(j, A')$. Therefore, we can see that $cost_A(\vec{t};\vec{x}_{-i},A) < cost_{A'}(\vec{t};\vec{x}_{-i},A)$, but $GM^*(\vec{t};\vec{x}_{-i},A) \neq A'$, this is a contradiction. Therefore, we have that $A = A'$. Recall that $O^*$ is lower numbered than $A$, when $k = 2, \sum_{j \in T} d(j, O^*) = \sum_{j \in T} d(j, A)$, we would have $cost_{O^*}(\vec{t};\vec{x}_{-i},A) = cost_A(\vec{t};\vec{x}_{-i},A)$. Hence, $GM^*(\vec{t};\vec{x}_{-i},A) \neq A = A'$ since $O^*$ is lower numbered.  
    
    Therefore, we can conclude that no client can alter the outcome of $GM^*$ by altering their reported location $x_i$ to $x_i'$. Hence, $\vec{x} = (O^*, O^*, \cdots, O^*)$ is a Nash equilibrium.
\end{proof}

Now that we have shown that Nash equilibrium always exists for our setting, we proceed to analyze the quality of Nash equilibria. We first show that the price of stability is extremely small, very close to 1 for arbitrary $k$, as long as $n$ is reasonably large. For example, if $n=100$, then the price of stability is only about $1.04$. This shows that if we can choose a Nash equilibrium (e.g., by choosing starting strategies for the clients, or by choosing the starting winning location wisely, but allowing the clients to change their report), then we can always form a solution which is essentially optimal.

\begin{theorem}
	Using Mechanism $GM^*$, for any finite $\F$, the price of stability is at most $\frac{n+1}{n-1}$ when $k = 1$ and is at most $\frac{n+\frac{k}{k-1}}{n -\frac{k}{k-1}}$ when $k > 1$. Thus, for arbitrary $k$, the price of stability is at most $\frac{n+2}{n-2}$.
	\label{thm:gen_pos}
\end{theorem}

\begin{proof}
	Let $O$ be the true optimal solution. If $O$ is the winner of a Nash equilibrium, then we are already done, so assume it is not a Nash equilibrium. We define a process of finding a Nash equilibrium as follows. We start with $A_1 = O$. At step $p$, set $\vec{x}^p = (A_p, A_p, \cdots, A_p)$ to be the vector of reported locations of all clients in $L$. If for all $i \in L$ and any location $x'_i$ that $i$ can report, $d(GM^*(\vec{t};\vec{x^p}_{-i}, x_i'),i) \geq d(GM^*(\vec{t};\vec{x^p}),i)$, then we are done, since this implies $\vec{x^p}$ is a Nash equilibrium, and we stop the process. Otherwise, find a client $i$ such that $d(GM^*(\vec{t};\vec{x^p}_{-i}, x_i'),i) < d(GM^*(\vec{t};\vec{x^p}),i)$, let the set of all such $x_i'$ be $S$, $x''_i = \text{argmin}_{x_i' \in S}d(GM^*(\vec{t};\vec{x^p}_{-i}, x_i'),i)$. If there is a tie, we choose $x_i'$ that results in facility location $GM^*(\vec{t};\vec{x^p}_{-i}, x'_i)$ with the lowest numbering. Let $A_{p+1} = GM^*(\vec{t};\vec{x^p}_{-i}, x_i'')$. Then, set $\vec{x}^{p+1} = (A_{p+1}, A_{p+1}, \cdots, A_{p+1})$ and repeat the process.  
    
    Observe that for all $p$, we have that $A_{p} = \text{argmin}_{A \in \F}cost_A(\vec{t};\vec{x^p}).$ For $p>1,$ this is true due to Lemma \ref{lemma:lie-on-top}: since $A_p$ has the smallest cost using locations $(\vec{t};\vec{x}^{p-1}_{-i}, x_i'')$, then it will still have the smallest cost after moving the reported locations on top of $A_p$ in $(\vec{t};\vec{x}^{p})$. For $p=1$, recall that $A_1=O$ is the true optimal solution. This means $O = \text{argmin}_{B \in \F}\sum_{i\in \C}d(i,B).$
	With all the liars reporting at $O$, $\vec{x}^1 = (O, O, \cdots, O)$, it is obvious that $O = \text{argmin}_{B \in \F}\sum_{i\in \C}d(x_i^1,B) $ for the first step of the process. 
	
	% Suppose all clients in $L$ are at location $x$ such that $GM^*(\vec{t};\vec{x})=x$, and it is a Nash equilibrium, then we are done. However, if it is not a Nash equilibrium, all clients in $L$ deviate to some location $y$ such that $GM^*(\vec{t};\vec{x})=y$. 
	
	If there exists no cycle in the process, it must terminate and result in a Nash equilibrium since there is a finite amount of facility locations in $\F$. 
	We will show that there doesn't exist a $A_q$ where $A_q = A_p, p < q$, and thus there is no cycle. 
    %Note that at each step, there must exist some $i \in L$ such that $d(i, GM^*(\vec{t};\vec{y}_{p-1})) > d(i, GM^*(\vec{t};\vec{y_{p-1}}_{-i}, A_p))$ (or by construction, $d(i, A_{p-1}) > d(i, A_p), GM^*(\vec{t};\vec{y}_{p-1}) = A_{p-1}, GM^*(\vec{t};\vec{y_{p-1}}_{-i}, A_p) = A_p$) so $i$ would report to be at $A_p$ instead.
 %    % Note that each time there would be some $i$ deviating to $a_p$    
    By Lemma \ref{lemma:lie-on-top}, we have that since $A_{p+1} = GM^*(\vec{t};\vec{x^p}_{-i}, x_i'')$, then $A_{p+1} = \text{argmin}_{B \in \F}cost_B(\vec{t};\vec{x^p}_{-i}, A_{p+1})$.
	% so that $GM^*(\vec{t};\vec{x}_{-i},a_p) = a_p$. 
	Then, we have that
	
	\[
	\begin{aligned}
        cost_{A_p}(\vec{t};\vec{x}^p_{-i},A_{p+1}) &\geq cost_{A_{p+1}}(\vec{t};\vec{x}^p_{-i},A_{p+1})\\
		\sum_{j \in T}d(j,A_p) + \sum_{j \in L\setminus\{i\}}d(A_p,A_p) + d(A_{p+1},A_p) &\geq \sum_{j \in T}d(j,A_{p+1}) + \sum_{j \in L\setminus\{i\}}d(A_p,A_{p+1}) + d(A_{p+1},A_{p+1}) \\
        \end{aligned}
        \]
    
        \begin{equation}   
		\sum_{j \in T}d(j,A_p)  \geq \sum_{j \in T}d(j,A_{p+1})+ (\numl-2)d(A_p,A_{p+1}) 
         \label{equ:geq}
	     \end{equation}

	for all $p$.
%	Note that this argument is true for all $a_p, a_{p+1}$, so we have 
%	$$\sum_{j \in T}d(j,a_p)  > \sum_{j \in T}d(j,a_{p+1}) + (\numl-2)d(o,a_{p+1}) $$
%	and 
%	$$\sum_{j \in T}d(j,a_m)  > \sum_{j \in T}d(j,a) + (\numl-2)d(o,a) $$
	Now, suppose to the contrary that there exists some $A_q$ such that $A_q = A_p$ but $p < q$. We first consider the case when $k > 2$. Repeatedly applying Inequality \ref{equ:geq}, we obtain:
	\[
		\begin{aligned}
			\sum_{j \in T}d(j,A_p)  &\geq \sum_{j \in T}d(j,A_{p+1}) + (\numl-2)d(A_p,A_{p+1})\\
			&\geq \sum_{j \in T}d(j,A_{p+2}) + (\numl-2)d(A_{p+1},A_{p+2}) +(\numl-2)d(A_p,A_{p+1})\\
			&\geq \cdots\\
			&\geq \sum_{j \in T}d(j,A_{q}) + (\numl-2)(d(A_p,A_{p+1})
        d(A_{p+1},A_{p+2})+ \cdots + d(A_{q-1}, A_q)\\
			&> \sum_{j \in T}d(j,A_p),
		\end{aligned}
	\]
	which is a contradiction. 
    
    We then consider the case when $k = 1$. Similar to the proof for Theorem \ref{thm:neexist}, we choose $A_2$ to be the location which minimizes $d(GM^*(\vec{t};A_2),i)$, i.e., the location which, if reported by $i$, results in the smallest cost for $i$ under the mechanism $GM^*$. Then $A_2$ is trivially a Nash equilibrium, and there is no cycle. 
    
    Now, we consider the case $k = 2$. Applying Inequality \ref{equ:geq} for $k=2$, we would have $\sum_{j \in T}d(j,A_p)  \geq \sum_{j \in T}d(j,A_{p+1}) \geq \sum_{j \in T}d(j,A_{p+2}) \geq \cdots \geq \sum_{j \in T}d(j,A_{q}) =  \sum_{j \in T}d(j,A_{p})$. If any of the inequalities is strict, then we are done and have a contradiction. Therefore, we only consider the case $\sum_{j \in T}d(j,A_p) = \sum_{j \in T}d(j,A_{p+1}) = \sum_{j \in T}d(j,A_{p+2}) = \cdots = \sum_{j \in T}d(j,A_{q}) =  \sum_{j \in T}d(j,A_{p})$. This implies that $cost_{A_{p+l}}(\vec{t};\vec{x}^{p+l}_{-i},A_{p+l+1}) = cost_{A_{p+l+1}}(\vec{t};\vec{x}^{p+l}_{-i},A_{p+l+1})$ for all $0 \leq l \leq q-p-1$. By construction of how $A_{p+l+1}$ is chosen, we must have that $A_{p+l+1}$ is lower numbered than $A_{p+l}$, meaning that $A_q$ must be lower numbered than $A_p$, but $p = q$, a contradiction. Therefore, such a process must exist without a cycle.

    Combining the results above, we have shown that our process must terminate for all $k \geq 1$, and results in a Nash equilibrium. Let $A_t$ be the last facility chosen, so we know that $\vec{x} = (A_t, A_t, \cdots, A_t)$ is a Nash equilibrium. We want to show that the cost of this solution is close to the cost of the optimum solution.
    
    By Inequality \ref{equ:geq}, we can see that 
	
		\begin{align}
			\sum_{j \in T}d(j,A_1)  &\geq \sum_{j \in T}d(j,A_{2}) + (\numl-2)d(A_1,A_{2})\nonumber\\
			&\geq \sum_{j \in T}d(j,A_{3}) + (\numl-2)d(A_2,A_{3}) +(\numl-2)d(A_1,A_{2})\nonumber\\
			&\geq \cdots\nonumber\\
			&\geq \sum_{j \in T}d(j,A_t) + (\numl-2)(d(A_1,A_{2})+d(A_2,A_{3}) + \cdots + d(A_{t-2}, A_{t-1}) + d(A_{t-1},A_t))\label{eqn.truthbefore}%\\
            %&\geq \sum_{j \in T}d(j,A_t) + (\numl-2)d(A_1,A_t)\label{eqn.truth}
        \end{align}

	%Note that since $GM^*(\vec{t};\vec{x}_{-i},a_p) = a_p$, 
	The last inequality holds by triangle inequality. Note that Inequality \ref{eqn.truthbefore} holds for $k=1$ as well, since $A_t=A_2$ in that case. By the construction of the process that chose $A_p$, for each $p$ there must exist some $i$ such that $d(i, A_p) < d(i, A_{p-1})$. Then, by triangle inequality, we have that 
	\[ 
	\begin{aligned}
		\sum_{j \in L}d(j,A_p) &\leq d(i, A_{p-1}) + \sum_{j\in L \setminus\{i\}}\left(d(j,A_{p-1})+d(A_{p-1},A_p)\right)\\
		&= \sum_{j\in L }d(j,A_{p-1})+ (k-1)d(A_{p-1},A_p),
	\end{aligned}
	\]
	for all $p$. So we have that 
	\[
	\begin{aligned}
		\sum_{j \in L}d(j,A_t) &\leq  \sum_{j\in L }d(j,A_{t-1})+ (k-1)d(A_{t-1},A_t) \\
		&\leq \sum_{j\in L }d(j,A_{t-2})+ (k-1)(d(A_{t-2},A_{t-1})+d(A_{t-1},A_t)) \\
		&\leq \cdots\\
			&\leq \sum_{j \in L}d(j,A_1) + (\numl-1)(d(A_1,A_{2})+d(A_2,A_{3}) + \cdots + d(A_{t-2}, A_{t-1}) + d(A_{t-1},A_t))\\
	\end{aligned}
	\]
	
	For simplicity of notation, let $A = A_t$. Also recall that $A_1 = O$. Thus, we have that 
    % \begin{equation}
    % \label{equ:l}
    %     \sum_{j \in L}d(j,A_p) \leq  \sum_{j\in L }d(j,A_{p-1})+ (k-1)d(A_{p-1},A_p), \forall p
    % \end{equation}
    
    \begin{equation}
        \sum_{j \in L}d(j,A) \leq \sum_{j \in L}d(j,O) + (\numl-1)(d(O,A_{2})+d(A_2,A_{3}) + \cdots + d(A_{t-2}, A_{t-1}) + d(A_{t-1},A))
        \label{equ:t}
    \end{equation}
    
    Now, we can analyze the quality of Nash equilibrium $A$ w.r.t. the true optimal solution. 
    
    First, we consider the case $k = 1$. %, recall that all the clients and facilities are in some metric space, we can utilize triangle inequality and $\sum_{j \in T}d(j,O) \geq \sum_{j \in T}d(j,A) + (\numl-2)d(O,A) $. In addition, we note that since $k = 1$, by the construction of how we obtain the Nash equilibrium, for $i \in L$, it must be that $d(i, A) < d(i,O)$. 
	Recall that $A_t=A_2$ for $k=1$, and so 
    Inequality \ref{equ:geq} states that $\sum_{j \in T}d(j,A_1)  \geq \sum_{j \in T}d(j,A_{t})+ (\numl-2)d(A_1,A_{t})$, this means that $\sum_{j \in T}d(j,A) \leq \sum_{j \in T}d(j,O) - (\numl-2)d(O,A)$. Now, with $k=1$, $L=\{i\}$, and by the definition of our process, we must have $d(i,A) < d(i,O)$ as $i$ prefers $A$ over $O$. This means that $\sum_{j \in T}d(j,A) + d(i,A) < \sum_{j \in T}d(j,O) - (\numl-2)d(O,A) + d(i,O)$. Now, with triangle inequality, we can obtain the following result for the cost of $A$: 
    % when $k=1$, with $L=\{i\}$: 
	\[
	\begin{aligned}
		\sum_{j \in T}d(j,A) + d(i,A) &< \sum_{j \in T}d(j,O) - (\numl-2)d(O,A) + d(i,O)  \\
		&=   \sum_{j \in \C}d(j,O) - (1-2)d(O,A)  \\
		&=   \sum_{j \in \C}d(j,O) + d(O,A)  \\
		&\leq  \sum_{j \in \C}d(j,O) + \frac{1}{n} \sum_{j \in \C}d(j,O) + \frac{1}{n} \sum_{j \in \C}d(j,A)                             
	\end{aligned}
	\]
	Note that the last inequality holds since $d(O,A) \leq d(j,O) + d(j,A)$ for all $j \in \C$. Now, rearrange the above inequality, and we have 
	\[
	\begin{aligned}
		\frac{n-1}{n} \sum_{j \in \C}d(j,A) &< \frac{n+1}{n}\sum_{j \in \C}d(j,O)\\
		\sum_{j \in \C}d(j,A) &< \frac{n+1}{n-1} \sum_{j \in \C}d(j,O) 
	\end{aligned}
	\]
	as desired. 
    
    Now, we consider the case where $k > 1$. Let $\delta = (d(O,A_{2})+d(A_2,A_{3}) + $$ \cdots + d(A_{t-1}, A))$, recall that we have $\sum_{j \in T}d(j,O) \geq \sum_{j \in T}d(j,A) + (\numl-2)\delta$ by Inequality \ref{eqn.truthbefore}, this means that $ \sum_{j \in T}d(j,A) \leq  \sum_{j \in T}d(j,O) - (\numl-2)\delta$. Then, we can see that 

    \[
	\begin{aligned}
		\sum_{j \in T}d(j,A) + \sum_{j \in L}d(j,A) &\leq \sum_{j \in T}d(j,O) - (\numl-2)\delta + \sum_{j \in L}d(j,A)\\
        &\leq \sum_{j \in T}d(j,O) - (\numl-2)\delta + \frac{1}{k-1}\sum_{j \in L}d(j,A) + \frac{k-2}{k-1} \sum_{j \in L}d(j,A)      
	\end{aligned}
	\]
    Now, by Inequality \ref{equ:t}, we have $\sum_{j \in L}d(j,A) \leq \sum_{j \in L}d(j,O) + (\numl-1)\delta$. Therefore, with triangle inequality, we have
\[
	\begin{aligned}
		\sum_{j \in T}d(j,A) + \sum_{j \in L}d(j,A) 
		&\leq \sum_{j \in T}d(j,O) - (\numl-2)\delta + \frac{1}{k-1}\sum_{j \in L}d(j,A)\\
        &\quad+ \frac{k-2}{k-1} \left(\sum_{j \in L}d(j,O) + (\numl-1)\delta\right)  \\
        &= \sum_{j \in T}d(j,O) - (\numl-2)\delta + \frac{1}{k-1}\sum_{j \in L}d(j,A) + \frac{k-2}{k-1} \sum_{j \in L}d(j,O) + (\numl-2)\delta\\
        &= \sum_{j \in T}d(j,O) + \frac{1}{k-1}\sum_{j \in L}d(j,A) + \frac{k-2}{k-1} \sum_{j \in L}d(j,O)            
	\end{aligned}
	\]
    Here we note that by triangle inequality, $d(j,A) \leq d(j,O) + d(O,A)$ for all $j \in L$. We can then obtain
    \[
\begin{aligned}
		\sum_{j \in T}d(j,A) + \sum_{j \in L}d(j,A) 
		&\leq \sum_{j \in T}d(j,O) + \frac{1}{k-1}\sum_{j \in L}d(j,O) + \frac{1}{k-1}\sum_{j \in L}d(O,A)+ \frac{k-2}{k-1} \sum_{j \in L}d(j,O)   \\
        &= \sum_{j \in T}d(j,O) + \sum_{j \in L}d(j,O) + \frac{1}{k-1}\sum_{j \in L}d(O,A)  \\
		&=   \sum_{j \in \C}d(j,O) + \frac{k}{k-1} \cdot d(O,A) \\
		&\leq  \sum_{j \in \C}d(j,O) + \frac{k}{n(k-1)} \sum_{j \in \C}d(j,O) + \frac{k}{n(k-1)} \sum_{j \in \C}d(j,A)             
	\end{aligned}
	\]
    
	The last inequality above holds since $d(O,A) \leq d(j,O) + d(j,A)$ for all $j \in \C$. Now, rearrange the above inequality, and we have 
	\[
	\begin{aligned}
		\frac{n-k/(k-1)}{n} \sum_{j \in \C}d(j,A) &\leq \frac{n+k/(k-1)}{n}\sum_{j \in \C}d(j,O)\\
		\sum_{j \in \C}d(j,A) &\leq \frac{n+k/(k-1)}{n-k/(k-1)} \sum_{j \in \C}d(j,O) 
	\end{aligned}
	\]
	
	Therefore, we can conclude that the price of stability is at most $\frac{n+1}{n-1}$ when $k = 1$ and is at most $\frac{n+k/(k-1)}{n -k/(k-1)}$ when $k > 1$.
\end{proof}

%We note that when $n \rightarrow \infty$, the upper bound for the price of stability is very close to 1.
The above bound is also almost tight for $k \leq n/2$ when $k > 2$:
 \begin{theorem}
 	Using Mechanism $GM^*$, for any $k$ with  $2< k \leq \frac{n}{2}$, there exists an instance where the best Nash equilibrium is a $\frac{n+1}{n-1} - \epsilon$ approximation (for any $\epsilon > 0$) of the true optimal solution, even on a line metric.
 	\label{thm:postight}
 \end{theorem}

\begin{proof}
 	% \yue{draw the graph}
\begin{figure}[!ht]
\centering
\begin{tikzpicture}[node distance=2cm, every node/.style={scale=1.1}]
    % Nodes
    \node[draw, rectangle, fill=cyan!30] (O) {$O$};
    \node[draw, circle, fill=cyan!30, right= of O,xshift=1.5cm] (A) {$A$};
    \node[draw, circle, fill=cyan!30, right= of A, xshift=-0.5cm] (B) {$B$};
    \node[draw, rectangle, fill=green!30, right= of B, xshift=0cm] (Ostar) {$O^*$};

    % Connections
    \draw (O) -- (A) node[midway, above] {1 + $\epsilon$};
    \draw (A) -- (B) node[midway, above] {\(\frac{k-2}{n-k}-\epsilon\)};
    \draw (B) -- (Ostar) node[midway, above] {\(\frac{n-2k+2}{n-k}\)};

    % Client and facility labels
    \node[above=0.3cm of O] {\(k-1\) clients, 1 facility};
    \node[below=0.3cm of A] {1 client};
    \node[below=0.3cm of B] {\(n-k\) clients};
    \node[above=0.3cm of Ostar] {1 facility};
\end{tikzpicture}
\caption{An instance of facility location problem where the price of stability bound is almost tight on a line. There are $k-1$ clients in $L$ at location $O$, the remaining 1 client in $L$ is at location $A$. All the $n-k$ clients in $T$ are at location B. There are also two possible facility locations $O$ and $O^*$ and $O^*$ is lower numbered than $O$, with the distances between two adjacent locations labeled.}
\label{fig:pos}
\end{figure}
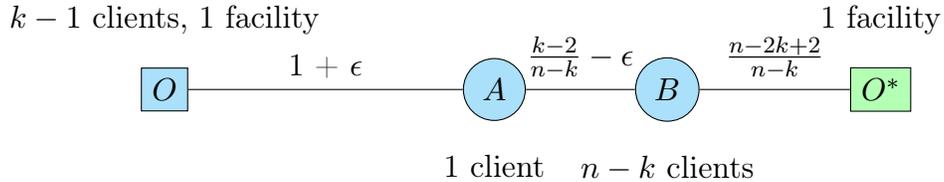

Consider the instance in the figure above. There are $k-1$ clients in $L$ at location $O$, the remaining 1 client in $L$ is at location $A$. All the $n-k$ clients in $T$ are at location $B$. There are also two possible facility locations $O$ and $O^*$, and $O^*$ is lower numbered than $O$, meaning it would be selected in case of a tie. The distances are as labeled, with $\epsilon > 0$ and sufficiently small. These are the true locations of the clients. We can see that choosing $O$ gives us a total cost of $n-1+\epsilon$ while $O^*$ gives us a total of $n+1-\epsilon$. This means that $O$ is the optimal solution. However, if the client at location $A$ reports to be at location $O^*$, the total perceived cost of choosing $O$ is $n$ but the perceived cost for choosing $O^*$ is also $n$, meaning $O^*$ would be chosen since it is lower numbered (note that the client who deviates prefers $O^*$ over $O$), therefore choosing $O$ is not an equilibrium under $GM^*$. 
Thus, under $GM^*$, all existing NE must choose $O^*$ (e.g., if all clients in $L$ report being at $O^*$, this is a NE). 
%Thus, under $GM^*$, only choosing $O^*$ is a Nash equilibrium, and not $O$. 
Therefore, the price of stability in this instance is $(n+1-\epsilon)/(n-1+\epsilon)$. With $\epsilon \rightarrow 0$, this value approaches $\frac{n+1}{n-1}$ as desired. 
 \end{proof}
 
 Next, we will evaluate the quality of the worst Nash equilibrium w.r.t. the true optimal solution. We present a (almost) tight bound on the quality of all Nash equilibria below. Note that for general metric spaces, this gives a bound of a small constant even for $k$ being about 40 percent of $n$, while the best strategyproof mechanisms cannot form a constant approximation for general metric spaces. 
 %Even for line metrics, for which simple strategyproof mechanisms provide an approximation ratio of 3, the Nash equilibria of $GM^*$ can be significantly better when $k$ is small. For example, when $k$ is about 15 percent of $n$, our bound gives a price of anarchy of only 1.86. 
 Thus, if we care about the cost of the resulting solution, and believe that the number of strategic agents is not too large, it makes sense to use our mechanism. This, of course, comes with the price of no longer having strategyproofness.

\begin{theorem}
 	Using Mechanism $GM^*$, the price of anarchy is at most $\frac{n+2k}{n-2k}$ for $ k < \frac{n}{2}$.
 	\label{thm:poa_gen}
 \end{theorem}
 \begin{proof}
Let $O$ be the true optimal solution, $W$ be the winner of a Nash equilibrium, and $x_i$ be the reported location of client $i \in L$ in this Nash equilibrium, so $W = \text{argmin}_{B \in \F}(\sum_{i\in T}d(i,B)+\sum_{i\in L}d(x_i,B))$. As usual, we denote the true location of client $i$ by $i$. Note that if $O = W$, we are done. So we assume that $O \neq W$.

 	%We consider the case where $k < \frac{n}{2}$. 
    Since $W$ is chosen by $GM^*$ given reported locations $x_i$, %$W = \text{argmin}_{B \in \F}\sum_{i\in \C}d(x_i,B)$, 
    we have that  \[
 	\begin{aligned}
    cost_W(\vec{t};\vec{x})&\leq cost_O(\vec{t};\vec{x})\\
 \sum_{i \in T} d(i, W)+\sum_{i \in L} d\left(x_i, W\right) &\leq \sum_{i \in T} d(i, O)+\sum_{i\in L} d\left(x_i, O\right) \\
\sum_{i \in T} d(i, W) &\leq  \sum_{i\in T} d(i, O)+\sum_{i \in L} d\left(x_i, O\right)-\sum_{i \in L} d\left(x_i, W\right) \\
&\leq  \sum_{i\in T} d(i, O)+\sum_{i \in L} d\left(x_i, W\right)+ \sum_{i \in L} d\left(O, W\right) -\sum_{i \in L} d\left(x_i, W\right) \\
& \leq  \sum_{i \in T} d(i, O)+\sum_{i \in L} d(W, O)
\end{aligned}
 	\]
 	The last two inequalities hold due to triangle inequality. Then, we can evaluate the Nash equilibrium with $W$ as its winner w.r.t. the true optimal solution and the true (rather than reported) locations.

    \[
    \begin{aligned}
\sum_{i \in T} d(i, W)+\sum_{i \in L} d(i, W) & \leq \sum_{i \in T} d(i, O)+\sum_{i \in L} d(W, O)+\sum_{i \in L} d(i, W) \\
& \leq \sum_{i \in T} d(i, O)+\sum_{i \in L} d(W, O)+\sum_{i \in L} d(i, O)+\sum_{i \in L} d(W, O) \\
& =\sum_{i \in \mathcal{C}} d(i, O)+2 k \cdot d(W, O) \\
& \leq \sum_{i \in \mathcal{C}} d(i, O)+\frac{2 k}{n} \sum_{i \in \mathcal{C}} d(i, W)+\frac{2 k}{n} \sum_{i \in \mathcal{C}} d(i, O) \\
& =\frac{n+2 k}{n} \sum_{i \in \mathcal{C}} d(i, O)+\frac{2 k}{n} \sum_{i \in \mathcal{C}} d(i, W)
\end{aligned}
    \]
    Now, rearrange the above inequality, and we have
    \[
    \begin{aligned}
\frac{n-2 k}{n} \sum_{i \in \mathcal{C}} d(i, W) & \leq \frac{n+2 k}{n} \sum_{i \in \mathcal{C}} d(i, O) \\
\sum_{i \in \mathcal{C}} d(i, W) & \leq \frac{n+2 k}{n-2 k} \sum_{i \in \mathcal{C}} d(i, O)
\end{aligned}
    \]
     as desired. 
 \end{proof}

Here, we note that the upper bound for POA can be further improved if all clients and facilities are located on a line.
 
 \begin{theorem}
 	Using Mechanism $GM^*$, if all clients and facilities are located on a line, the price of anarchy is at most $\frac{n+1}{n-1}$ when $k = 1$ and at most $\frac{n+2k-2}{n-2k+2}$ for $ 1 < k < \frac{n}{2}$.
 	\label{thm:poa}
 \end{theorem}
 \begin{proof}
 	Let $O$ be the true optimal solution, $W$ be the winner of a Nash equilibrium, and $x_i$ be the reported location of client $i \in L$ in this Nash equilibrium, so $W = \text{argmin}_{B \in \F}(\sum_{i\in T}d(i,B)+\sum_{i\in L}d(x_i,B))$. As usual, we denote the true location of client $i$ by $i$. Note that if $O = W$, we are done. So we assume that $O \neq W$.

%  	%We consider the case where $k < \frac{n}{2}$. 
%     Since $W$ is chosen by $GM^*$ given reported locations $x_i$, %$W = \text{argmin}_{B \in \F}\sum_{i\in \C}d(x_i,B)$, 
%     we have that
%  	\[
%  	\begin{aligned}
%     cost_W(\vec{t};\vec{x})&\leq cost_O(\vec{t};\vec{x})\\
%  \sum_{i \in T} d(i, W)+\sum_{i \in L} d\left(x_i, W\right) &\leq \sum_{i \in T} d(i, O)+\sum_{i\in L} d\left(x_i, O\right) \\
% \sum_{i \in T} d(i, W) &\leq  \sum_{i\in T} d(i, O)+\sum_{i \in L} d\left(x_i, O\right)-\sum_{i \in L} d\left(x_i, W\right) \\
% &\leq  \sum_{i\in T} d(i, O)+\sum_{i \in L} d\left(x_i, W\right)+ \sum_{i \in L} d\left(O, W\right) -\sum_{i \in L} d\left(x_i, W\right) \\
% & \leq  \sum_{i \in T} d(i, O)+\sum_{i \in L} d(W, O)
% \end{aligned}
%  	\]
%  	The last two inequalities hold due to triangle inequality.
By Theorem \ref{thm:poa_gen}, we have that $\sum_{i \in T} d(i, W) \leq \sum_{i \in T} d(i, O)+\sum_{i \in L} d(W, O)$. Then, we can evaluate the Nash equilibrium with $W$ as its winner w.r.t. the true optimal solution and the true (rather than reported) locations. We will consider two cases: (i) for all $j \in L$, $d(j, O) \geq d(j,W)$, (ii) there exists some $j \in L$ such that $d(j, O) < d(j,W)$. We first consider the case where for all $j \in L$, $d(j, O) \geq d(j,W)$.
 	\[
 	\begin{aligned}
 		\sum_{i \in \C}d(i, W) = \sum_{i \in T}d(i,W) + \sum_{i \in L}d(i,W) &\leq \sum_{i \in T} d(i, O)+\sum_{i \in L} d(W, O) + \sum_{i \in L}d(i,W)\\
 		&\leq \sum_{i \in T} d(i, O)+\sum_{i \in L} d(W, O) + \sum_{i \in L}d(i,O)\\
 		&= \sum_{i \in \C} d(i, O) + k \cdot d(W,O)\\
 		&\leq \sum_{i \in \C} d(i, O) + \frac{k}{n} \sum_{i \in \C}d(i,W) + \frac{k}{n} \sum_{i \in \C}d(i,O)\\
 		&= \frac{n+k}{n}\sum_{i \in \C} d(i, O) + \frac{k}{n} \sum_{i \in \C}d(i,W)
 	\end{aligned}
 	\]
 	The last inequality holds since $d(W,O) \leq d(i,W) + d(i,O)$ for all $i \in \C$ (triangle inequality). Now, rearrange the above inequality, and we have that $\sum_{i \in \C}d(i,W) \leq \frac{n+k}{n-k} \sum_{i \in \C}d(i,O)$ for this case.
	% \[
	% \begin{aligned}
	% 	\frac{n-k}{n} \sum_{i \in \C}d(i,W) &\leq \frac{n+k}{n}\sum_{i \in \C}d(i,O)\\
	% 	\sum_{i \in \C}d(i,W) &\leq \frac{n+k}{n-k} \sum_{i \in \C}d(i,O) 
	% \end{aligned}
	% \]

We then consider the case where there exists some $j \in L$ such that $d(j, O) < d(j,W)$. Let $x_j'=O$, and let $A = GM^*(\vec{t};\vec{x}_{-j},x_j')$. Then by Lemma \ref{lemma:lineshift}, we must have either $A = W$ or $d(j, A) < d(j, W)$. In other words, by having $j$ switch their reported location to $O$, we either still end up with $W$ as the winner, or end up with something that $j$ prefers to $W$. However, since $W$ is the winner at NE, it cannot be that $d(j, A) < d(j, W)$, otherwise $j$ could deviate to $O$ and improve their cost, contradicting that we are in a NE. Thus, it must be that $W = A = GM^*(\vec{t};\vec{x}_{-j},O)$, meaning that $W$ is selected over $O$ by our mechanism given reported locations $(\vec{t};\vec{x}_{-j},O)$. Thus,

	\[
	\begin{aligned}
    cost_W(\vec{t};\vec{x}_{-j},O) &\leq cost_O(\vec{t};\vec{x}_{-j},O)\\
		\sum_{i \in T}d(i,W) + d(O,W) + \sum_{i \in L\setminus \{j\}}d(x_i,W) &\leq \sum_{i \in T}d(i,O) + d(O,O) + \sum_{i \in L\setminus \{j\}}d(x_i,O)\\
		\sum_{i \in T}d(i,W)   &\leq \sum_{i \in T}d(i,O)  + \sum_{i \in L\setminus \{j\}}d(x_i,O)\\
        &\quad - \sum_{i \in L\setminus \{j\}}d(x_i,W) - d(O,W)
	\end{aligned}
	\]
	Now, by triangle inequality, we have that
	\[
	\begin{aligned}
		\sum_{i \in T}d(i,W) &\leq \sum_{i \in T}d(i,O)  + \sum_{i \in L\setminus \{j\}}d(x_i,W) +  \sum_{i \in L\setminus \{j\}}d(W,O) - \sum_{i \in L\setminus \{j\}}d(x_i,W) - d(O,W)\\
		&= \sum_{i \in T}d(i,O)  + \sum_{i \in L\setminus \{j\}}d(W,O)  -  d(O,W)\\
		   %&= \sum_{i \in T}d(i,O)  + (k-1)\cdot d(W,O)  -  d(O,W)\\
		   &= \sum_{i \in T}d(i,O)  + (k-2)\cdot d(W,O) 
	\end{aligned}
	\]
	Then, we can see that
	\[
	\begin{aligned}
		\sum_{i \in T} d(i,W) + \sum_{i \in L} d(i,W) &\leq \sum_{i \in T}d(i,O)  + (k-2)\cdot d(W,O) + \sum_{i \in L} d(i,W) \\
	&\leq \sum_{i \in T}d(i,O)  + (k-2)\cdot d(W,O) + \sum_{i \in L} d(i,O) + \sum_{i \in L} d(W,O)\\
		&= \sum_{i \in \C}d(i,O)  + (k-2)\cdot d(W,O) + k\cdot d(W,O)\\
		&= \sum_{i \in \C}d(i,O)  + (2k-2)\cdot d(W,O)\\
		&\leq \sum_{i \in \C} d(i, O) + \frac{2k-2}{n} \sum_{i \in \C}d(i,W) + \frac{2k-2}{n} \sum_{i \in \C}d(i,O)\\
 		&= \frac{n+(2k-2)}{n}\sum_{i \in \C} d(i, O) + \frac{2k-2}{n} \sum_{i \in \C}d(i,W) 
	\end{aligned}
	\]
	The last inequality holds since $d(W,O) \leq d(i, W) + d(i,O)$ for all $i \in \C$. Now, rearranging the above inequality, we have 
	\[
	\begin{aligned}
		\frac{n-(2k-2)}{n} \sum_{i \in \C}d(i,W) &\leq \frac{n+(2k - 2)}{n}\sum_{i \in \C}d(i,O)\\
		%\sum_{i \in \C}d(i,W) &\leq \frac{n+(2k-2)}{n-(2k-2)} \sum_{i \in \C}d(i,O) \\
		\sum_{i \in \C}d(i,W) &\leq \frac{n+2k-2}{n-2k+2} \sum_{i \in \C}d(i,O)
	\end{aligned}
	\]
Combining the results for the two cases, we note that $\frac{n+2k-2}{n-2k+2} > \frac{n+k}{n-k}$ for any $n > 0, k < \frac{n}{2}$, therefore the price of anarchy is at most $\frac{n+2k-2}{n-2k+2}$ for $k < \frac{n}{2}$ as desired. 
Combining the results for the two cases, for any $n > 0$, we note that $\frac{n+2k-2}{n-2k+2} \geq \frac{n+k}{n-k}$ for any $1 < k < \frac{n}{2}$, $\frac{n+2k-2}{n-2k+2} < \frac{n+k}{n-k}$ when $k=1$. Therefore, the price of anarchy is at most $\frac{n+1}{n-1}$ when $k = 1$ and is at most $\frac{n+2k-2}{n-2k+2}$ for $1 < k < \frac{n}{2}$ as desired.
 \end{proof}

The price of anarchy bound we gave above is tight for any anonymous monotonic mechanism. This means that there is no reason to use a different mechanism than $GM^*$ to improve the price of anarchy: any reasonable mechanism will have the same worst-case performance, so we might as well use the natural mechanism $GM^*$. We define {\em monotonic} mechanisms as a version of the definition in \cite{aziz2022strategyproof}, where the definition applied only to line metrics. In our work, a monotonic mechanism $M$ is a mechanism with the following property. Suppose $M(\vec{x})=A$, and some agent $i$ switches their reported location to be directly on top of $A$, thus expressing their extreme preference for $A$ being chosen. Then, $M(\vec{x}_{-i},A)=A$. In other words, someone increasing their preference for $A$ should not cause the mechanism to stop selecting $A$. An {\em anonymous} mechanism is simply one which does not use the identities of the clients to determine the outcome.
 
 \begin{theorem}
 	For any anonymous monotonic mechanism, there exists an instance where the worst Nash equilibrium is a $\frac{n+2k-2}{n-2k+2}$ approximation of the true optimal solution when $k < \frac{n}{2}$. For $k\geq \frac{n}{2}$, there are instances when the price of anarchy of any such mechanism is at least $\Omega(n)$. 
 	\label{thm:poatight}
 \end{theorem}

 \begin{proof}
\begin{figure}[h!]
\centering
\begin{tikzpicture}[node distance=4cm, every node/.style={scale=1}]
    % Nodes
    \node[draw, rectangle, fill=cyan!30] (O) {O};
    \node[draw, rectangle, fill=cyan!30, right=of O] (W) {W};

    % Connection
    \draw (O) -- (W) node[midway, above] {1};

    % Labels under nodes
    \node[below=0.5cm of O] {\(\frac{n}{2} - 1\) clients in \(T\)};
    \node[below=1cm of O] {\(k\) clients in \(L\)};
    
    \node[below=0.5cm of W] {\(\frac{n - 2k}{2} + 1\) clients in \(T\)};
\end{tikzpicture}
\caption{An instance of the facility location problem where there are two facility locations $O$ and $W$. All the liars' true locations are at $O$ and there are \(\frac{n}{2} - 1\) clients in \(T\) at $O$ and \(\frac{n - 2k}{2} + 1\) clients in \(T\) at $W$. $O$ and $W$ are distance 1 apart. }
\label{fig:poa}
\end{figure}
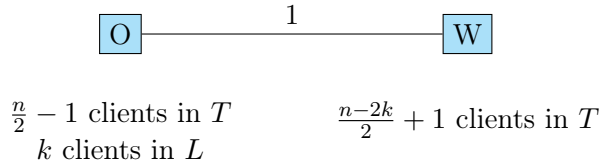

We consider the instance in Figure \ref{fig:poa}. There are two facility locations $O$ and $W$, with $W$ being lower-numbered than $O$. All the clients in $L$ are located at $O$, and there are \(\frac{n}{2} - 1\) clients in \(T\) at $O$ and \(\frac{n - 2k}{2} + 1\) clients in \(T\) at $W$. In addition, $d(O,W) = 1$. These are the true locations of the clients. We can see that choosing $O$ gives us a total cost of $\frac{n-2k}{2}+1$ while $W$ gives us a total of $\frac{n+2k}{2}-1$. This means that $O$ is the optimal solution. However, if we consider the case $\vec{x} = (W,W,\cdots, W)$, $(\vec{t};\vec{x})$ would be a Nash equilibrium and $W = GM^*(\vec{t};\vec{x})$ as more than half of the clients would be at $W$. Then, the price of anarchy would be $\frac{n + 2k - 2}{n -2k + 2}$ as desired. 

Next, we claim that such a bound holds for any monotonic mechanism $M$. Suppose $M$ is given an instance with two facility locations $A$ and $B$ with $d(A,B)=1$, with these being the only points in the metric space. Without loss of generality, assume that when $\frac{n}{2}$ clients report being on $A$ and $\frac{n}{2}$ report being on $B$, then $M$ chooses $A$ as the winner. Since $M$ is monotonic, then when $\frac{n}{2}+1$ clients report being on $A$ while $\frac{n}{2}-1$ report being on $B$, $M$ will still choose $A$ as the winner. This is a Nash equilibrium for $M$: any single client switching from $A$ to $B$ will result in $A$ still winning since $\frac{n}{2}$ clients will report being on top of $A$, and any client switching from $B$ to $A$ will result in $A$ winning due to monotonicity. Now consider the case when $A=W$ and $B=O$ in the example above, with the true client locations as in Figure \ref{fig:poa}. This gives us the price of anarchy bound of $\frac{n + 2k - 2}{n -2k + 2}$ for any anonymous monotonic mechanism when $k<\frac{n}{2}$.

% Next, we claim that such a bound holds for any monotonic mechanism. We first define monotonic mechanisms as follows. Given two clients' location profiles on a real line, $\vec{x} \in [0,1]^n, \vec{x'} \in [0,1]^n$, we say that $\vec{x} < \vec{x'}$ iff for all $i\in C
% $, $x_i \leq x'_i$ and for some $i\in C
% $, $x_i < x'_i$. Then, a mechanism $M$ is {\it monotonic} if $M(\vec{x}) \leq M(\vec{x'})$ for all $\vec{x} < \vec{x'}$ (note that this is defined on a line, modified from the definition of strictly monotonic mechanisms in \cite{aziz2022strategyproof}). Given any anonymous monotonic mechanism $M$, we consider the same example as shown in Figure \ref{fig:poa} with $\vec{x} = (W,W,\cdots, W)$, $(\vec{t};\vec{x})$. If $M(\vec{t};\vec{x}) = W$, then we are done and would get the same bound as previously discussed. Assume $M(\vec{t};\vec{x}) = O$, we note that there are $n/2 - 1$ clients at $O$ and $n/2 + 1$ clients at $W$ with reported location profile $\vec{x}$. Now, suppose two clients move from $W$ to $O$ so there are $n/2 + 1$ clients at $O$ and $n/2 - 1$ clients at $W$ with reported location profile $\vec{x'}$. Then $M(\vec{x'}) = W$ because $M$ is anonymous, but $W > O, \vec{x'} < \vec{x}$, meaning that $M$ is not monotonic, a contradiction. Therefore, we must have $M(\vec{t};\vec{x}) = W$. Hence, the above bound holds for any anonymous monotonic mechanism. 

We now consider the case $k \geq \frac{n}{2}$. We claim that the price of anarchy is $\Omega(n)$. To show this, consider the example where there are two locations $O, W$, with all clients' true location except one client in $T$ at $O$, and that one client's location is at $W$. Suppose that $W$ is preferred by the mechanism in case of ties. It is obvious that $O = \text{argmin}_{A\in \F}\sum_{i \in \C}d(i,A)$. Now, consider the case where all clients in $L$ are reporting to be at $W$ such that $\vec{x} = (W, W, \cdots, W)$. Since there are at least half of the clients reporting being at location $W$, then even if one of the clients in $L$ deviates to location $O$, this is the case where we have $(\vec{t};\vec{x}_{-i}, O)$, and it is easy to see that $GM^*(\vec{t};\vec{x}_{-i}, O) = W$ for any $i \in L$. This means that $\vec{x} = (W, W, \cdots, W)$ is a Nash equilibrium with $GM^*(\vec{t};\vec{x}) = W$. However, we note that the cost for choosing location $W$ is a factor of $n-1$ greater than choosing the optimal solution $O$. As $n \rightarrow \infty$, this becomes unbounded.

The same argument works for any anonymous monotonic mechanism $M$. Suppose again, as above, that $A$ is a Nash equilibrium for $M$. Then it could be that the true locations are as in the example in the previous paragraph, with $A=W$ and $B=O$, in which case the price of anarchy of $M$ would be $\Omega(n)$ when $k\geq \frac{n}{2}$.
 \end{proof}

  \section{Strong Nash Equilibrium}
  \label{sec:sne}
In this section, we consider the case when entire arbitrary coalitions of clients can collude and misreport their locations to gain a benefit. Thus, we study the existence and quality of strong Nash equilibrium, which is a solution stable against coalitional deviations. 
We have shown that Nash Equilibrium always exists in arbitrary metric spaces. However, this is not true for strong Nash equilibrium. 

\begin{theorem}
Strong Nash equilibrium does not always exist for Mechanism $GM^*$ in a general metric space.
	\label{thm:sne_gen}
\end{theorem}

\begin{proof}
\begin{figure}[h!]
\centering
\begin{tikzpicture}[every node/.style={scale=1.1}, node distance=2cm]

    % Facility nodes (green rectangles)
    \node[draw, rectangle, fill=green!30] (A) at (0,3) {A};
    \node[draw, rectangle, fill=green!30] (B) at (-2,0) {B};
    \node[draw, rectangle, fill=green!30] (C) at (2,0) {C};

    % Client nodes (blue circles)
    \node[draw, circle, fill=cyan!30] (i) at (-0.9,1.7) {i};
    \node[draw, circle, fill=cyan!30] (j) at (1.3,1.1) {j};
    \node[draw, circle, fill=cyan!30] (k) at (-0.4,0) {k};

    % Edges and distances
    \draw (A) -- (i) node[midway, left] {1};
    \draw (i) -- (B) node[midway, left] {2};
    \draw (A) -- (j) node[midway, right] {2};
    \draw (j) -- (C) node[midway, right] {1};
    \draw (B) -- (k) node[midway, below] {1};
    \draw (k) -- (C) node[midway, below] {2};

\end{tikzpicture}
\caption{An instance of a facility location problem with graph metric. There are three facility locations $A,B,C$ and three clients $i,j,k$ with the distance between them labeled.}
\label{fig:cycle}
\end{figure}
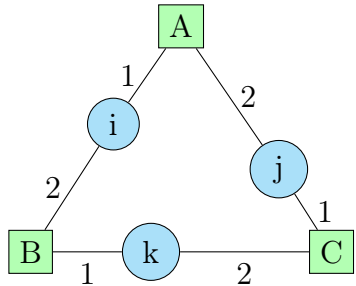
	We construct an example where no strong Nash equilibrium exists; this example is essentially a Condorcet cycle embedded into a metric space. As shown in Figure \ref{fig:cycle}, there are three facility locations $A,B,C$, and three clients $i,j,k\in L$, with $T=\emptyset$. We note that if $A$ is chosen, then $j$ and $k$ can misreport to locate at $C$ and alter the winner to be $C$. However, if $C$ is chosen, $i$ and $k$ can misreport to locate at $B$ and alter the winner to be $B$. Similarly, if $B$ is the chosen, $i$ and $j$ can misreport to locate at $A$ and alter the winner to be $A$. Therefore, there exists no strong Nash equilibrium. 
\end{proof}

% An example with no SNE is simple to construct: it is essentially the Condorcet cycle embedded into a metric space (see Appendix). 
However, if we restrict the metric space to a line, then there always exists a strong Nash equilibrium. 

\begin{theorem}
Every client reporting the facility location in $\F$ closest to them is a strong Nash equilibrium for Mechanism $GM^*$ on a line metric. 
	\label{thm:sne_line}
\end{theorem}

To see why this must be, we will first state a useful lemma. Note that this lemma is already known, and is implied in, e.g., \cite{feldman2016voting, anshelevich2017randomized}. We include the proof here for completeness. We will first define some useful terms.

\begin{definition}
    If all clients and facilities are located on a line, we define $Med_L$ and $Med_R$ as follows:
    \begin{enumerate}[(1)]
        \item If the number of clients is odd, let the true location of the median client be $i$, then $Med_L = Med_R = i$.
        \item If the number of clients is even, let the true location of the two median clients be $i$ and $j$, (with $i$ to the left or colocated with $j$), then $Med_L = i, Med_R = j$.
    \end{enumerate}
\end{definition}

\begin{lemma}
	If all clients and facilities are located on a line, if there are any facility locations in the interval $[Med_L, Med_R]$, then all such locations are optimal. Otherwise, the optimal facility location is either the closest facility location to the right of $Med_R$ or the closest facility location to the left of $Med_L$. 
	\label{lemma:optline}
\end{lemma}
\begin{proof}
 We first note that it is obvious that placing the facility at the interval $[Med_L, Med_R]$ is the optimal solution if any of such locations are in $\F$. Then, we assume these locations are not in $\F$. For the sake of contradiction, we assume the optimal facility location is neither the closest facility location to the right of $Med_R$ nor the closest facility location to the left of $Med_L$. Without loss of generality, assume the optimal location $O$ is to the left of the closest facility location $A$ to the left of $Med_L$. Let the number of clients to the right of $A$ be $r > 0$, the number of clients to the left of $O$ be $l_1 > 0$, the set clients in between $O$ and $A$ be  $S, |S| = l_2 > 0$. Since $A$ is the closest facility to the left of $Med_L$, we have $l_1 + l_2 < r$. Then, we have $cost_O(\vec{t};\vec{x}) - cost_A(\vec{t};\vec{x}) > (r-l_1)\cdot d(A,O) + \sum_{i \in S}\left(d(i, O) - d(i, A) \right) > (r-l_1 - l_2)\cdot d(A,O) > 0$. But this means that $O$ is not the optimal solution, a contradiction. %The other case is symmetric. %If all clients and facilities are located on a line, the optimal facility location is either the closest facility location to the right of the median client or the closest facility location to the left of the median client. If the closest location is the median client's location, then the optimal facility location is the median location. 
\end{proof}

Using this lemma, we can now prove Theorem \ref{thm:sne_line}. We do this by showing that if $Z$ is chosen by the mechanism, then it is not possible for any coalition of clients who all prefer $Y$ to $Z$ to misreport and cause $Y$ to win, since they would not be able to make $Y$ become a facility next to the median reported client. 

\begin{proof}

For any client $i\in L$, let $x_i$ ($i$'s reported location) be the facility location closest to them such that $x_i = \text{argmin}_{A\in \F}d(i,A)$. Let $Z = GM^*(\vec{t};\vec{x})$.  Now, we consider some location $Y$, located to the left of $Z$ (in the interval $(-\infty, Z)$). We will show that no group of clients who strictly prefer $Y$ can alter their reported location so that $GM^*$ would choose $Y$. We note that any clients to the right of $Z$ (in the interval $[Z, \infty)$) would prefer $Z$ over any location to the left of $Z$, and thus would not be part of such a deviating coalition.

%w.l.o.g, we consider clients in $L$ to the left of $Z$. 
%$$\Delta_z = \left(\sum_{j \in T} d(j, z) + \sum_{j \in L\setminus\{i\}} d(x_j, z) + d(x_i', z)\right) - \left(\sum_{j \in T} d(j, z) + \sum_{j \in L} d(x_j, z)\right)$$
%$$\Delta_y = \left(\sum_{j \in T} d(j, y) + \sum_{j \in L\setminus\{i\}} d(x_j, y) + d(x_i', y)\right) - \left(\sum_{j \in T} d(j, y) + \sum_{j \in L} d(x_j, y)\right)$$

We first consider the case where there exist clients in $L$ whose reported location is to the right of $Y$ (in the interval $(Y, \infty)$) and prefer $Y$ over $Z$. Call this set of clients $Y_R$. Recall that $Y_R$ does not contain clients to the right of $Z$, so all of the clients in $Y_R$ are in the interval $(Y, Z)$. Take the rightmost $x_i$ for $i \in Y_R$, denote it by $x_r$. Then, by the definition of $x_r$, only clients who reported in the interval $(-\infty, x_r]$ 
% to the left and on top of $x_r$ 
have an incentive to change their reported location to make $Y$ win. Since $Z = GM^*(\vec{t};\vec{x})$, by Lemma \ref{lemma:optline}, we have that the $Med_R$ of $(\vec{t};\vec{x})$ must be to the right of $x_r$ (in the interval $(x_r, \infty)$), otherwise $Z$ couldn't have won. Note that $Med_L$ must also be to the right of $x_r$ (in the interval $(x_r, \infty)$). To see why, assume to the contrary that $Med_L$ is on top of or to the left of $x_r$ (in the interval $(-\infty, x_r]$), which means that $x_r \in [Med_L, Med_R]$. By Lemma \ref{lemma:optline}, $x_r$ is one of the optimal solutions, but since $x_r$ is to the left of $Z$, it is lower numbered and thus $Z$ wouldn't have been chosen by $GM^*$, which is a contradiction. Thus, both $Med_R$ and $Med_L$ must be to the right of $x_r$ (in the interval $(x_r, \infty)$). This means that no matter how clients who prefer $Y$ over $Z$ alter their reported locations, $Med_L$ would still remain in $(x_r, \infty)$, since all such clients have reported $x_i$ in $(-\infty, x_r]$. By Lemma \ref{lemma:optline}, since $x_r \neq Y$ and is a facility location, then $Y$ cannot be the left closest facility ($x_r$ would always be closer) to $Med_L$ or in the interval $[Med_L, Med_R]$, and thus would not be chosen. Therefore, in this case no coalition which prefers $Y$ over $Z$ can change their reported locations and make $Y$ become the winner.

Next, assume that $Y_R = \emptyset$. Then all clients in $L$ who prefer $Y$ over $Z$ have reported locations to the left or on top of $Y$ (in the interval $(-\infty, Y]$). Call this set of clients $Y_L$. For each $i \in Y_L$, let $i$ deviate to $x_i'$. We will consider three cases: (i) $x_i'$ is to the left or on top of $Y$ (in the interval $(-\infty, Y]$), (ii) $x_i'$ is to the right of $Y$ but to the left or on top of $Z$ (in the interval $(Y, Z]$), (iii) $x_i'$ is to the right of $Z$ (in the interval $(Z, \infty)$). In addition, for any $i \in \C$ and $(\vec{t};\vec{x}_{-i}, x_i')$ for some location $x_i'$, let $\Delta_{Zi} = cost_Z(\vec{t};\vec{x}_{-i}, x_i') - cost_Z(\vec{t};\vec{x})$, $\Delta_{Yi} = cost_Y(\vec{t};\vec{x}_{-i}, x_i') - cost_Y(\vec{t};\vec{x})$.
We first consider case (i). Assume that $x_i'$ is at distance $\Delta$ to the left of $x_i$, we have $\Delta_{Zi} = \Delta_{Yi} = \Delta$.
% meaning that such deviation would preserve $GM^*(\vec{t};\vec{x}_{-i},x_i') = z$. 
If $x_i'$ is at distance $\Delta$ to the right of $x_i$, we have $\Delta_{Zi} = \Delta_{Yi} = -\Delta$. In either case, the change in cost for both $Z$ and $Y$ is the same.
Then, we consider case (ii), $x_i'$ is in the interval $ (Y, Z]$. In this case, we have 
\begin{eqnarray*}\Delta_{Zi} = d(Z,x_i') - d(x_i,Z)
= d(Z,x_i') - d(x_i,x_i')-d(x_i',Z)
= -d(x_i,x_i')\\
= -d(x_i,Y) - d(x_i',Y)
\leq -d(x_i,Y)
\leq d(x_i',Y) -d(x_i,Y)
= \Delta_{Yi}.
\end{eqnarray*}
Using a similar argument, we can see that $\Delta_{Zi}\leq \Delta_{Yi}$ also holds for case (iii), $x_i'$ is in the interval $(Z, \infty)$. Therefore, the total differences after all clients in $Y_L$ alter their reported location would have the property $\sum_{i\in Y_L}\Delta_{Zi} \leq \sum_{i\in Y_L}\Delta_{Yi}$.
Given that $GM^*(\vec{t};\vec{x}) = Z$, this means that the cost of choosing $Y$ does not become lower as compared to the cost of choosing $Z$ for $(\vec{t};\vec{x}_{-i},x_i')$, so $Z$ would still be chosen over $Y$. 
%This would either make $GM^*(\vec{t};\vec{x}_{-i},x_i') = Z$ or to the right of $z$. 
% Therefore, we can conclude that no $i$ would have the incentive to alter their reported location $x_i$ to the right of $Y$. Combined with the previous result, we can see that whatever a set of clients in $Z_L$ alter their reported location, $y$ would not be chosen.

% Then, we consider clients in $L$ to the right of $y$ who prefer $y$ over $z$. Call this set of clients $Y_R$. Recall that this $Y_R$ does not contain clients to the right of $z$. Take the rightmost $x_i$ for $i \in Y_R$, denote it by $x_r$. Since $z = GM^*(\vec{t};\vec{x})$, by Lemma \ref{lemma:optline}, we have that the median of $(\vec{t};\vec{x})$ must be to the left of $x_r$, otherwise $z$ couldn't have won. Now, since $x_r \neq y$, $y$ would not be chosen since $y$ is not the left or right closest facility to the median reported location, no matter how clients in $Y_R$ and $Y_L$ alter their reported location. 

Finally, we consider the case where $Y$ is located to the right of $Z$ (in the interval $(Z, \infty)$). This case is mostly symmetric to the case where $Y$ is located to the left of $Z$ (thus we can apply the same proof) except when there exist clients in $L$ whose reported location is to the left of $Y$ (in the interval $(-\infty, Y)$) and prefer $Y$ over $Z$. Call this set of clients $Y_L$. Note that $Y_L$ does not contain clients to the left of $Z$, so they are all in the interval $(Z, Y)$. Take the leftmost $x_i$ for $i \in Y_L$, denote it by $x_l$. Then, by the definition of $x_l$, only clients who reported in the interval $[x_l, \infty)$
%to the right and on top of $x_l$ 
have an incentive to change their reported location to make $Y$ win. By Lemma \ref{lemma:optline}, for $Z$ to be selected, it must be that $Med_L\in (-\infty, x_l)$. Thus we have either $x_l$ is to the right of both $Med_L$ and $Med_R$ (in the interval $(Med_R, \infty)$), or both $x_l$ and $Z$ are in the interval $[Med_L, Med_R]$. For the first case, the proof follows exactly as the case where $Y$ is located to the left of $Z$. For the second case, since $Med_L$ is either on top or to the left of $Z$ (in the interval $(-\infty, Z]$), no matter how clients who prefer $Y$ over $Z$ alter their reported location, $Z$ would still be either in the interval $[Med_L, Med_R]$  or to the right of $Med_R$ (in the interval $(Med_R, \infty)$) but closer to $Med_R$ than $Y$. This is because only clients with reported location in the interval $(Z, \infty)$ have an incentive to deviate, but $Med_L$ was in the interval $(-\infty, Z]$, so the new $Med_L$ after the deviation must also be in the same interval. Thus, by Lemma \ref{lemma:optline}, we know that $Y$ could not be chosen by $GM^*$ no matter how those who prefer $Y$ change their reported locations (for the former case where $Z$ is in the interval $[Med_L, Med_R]$ after the changes, it is also because $Z$ is lower numbered than $Y$ as it is to the left of $Y$).

	Since the choice of $Y$ is arbitrary, we can conclude that no group of clients can misreport and improve their costs. Hence, every client reporting their closest facility location is a strong Nash equilibrium on a line. 
\end{proof}

Next, we will analyze the quality of strong Nash equilibria. For the next results, let $O$ be the true optimal solution, and $X$ be the winner of an arbitrary strong Nash equilibrium. Let $L_O$ be the set of clients in $L$ that strictly prefer $O$ over $X$, where $|L_O| = p$, and $L_X$ be the set of clients in $L$ that prefer $X$ over $O$ (or are indifferent), where $|L_X| = k - p; L_X \cap L_O = \emptyset, L_X \cup L_O = L$. %In addition, for simplicity, we will denote the true location of client $i$ by $i$ and the reported location by $j \in L$ by $x_j$.

First, we will show a useful lemma. This lemma states the following. Suppose that some agents that prefer $B$ over the current winner $W$ report their location to be at $B$ instead. Then, the resulting winner may not be $B$ itself, but it will either be $W$ (the agents changing their reported location does not change the outcome), or it will be something which all the deviating agents prefer to $W$ (so it is in the best interest of these agents to switch their reported location to $B$). This lemma rules out the annoying possibility that, by switching their reported location to $B$, the result is an outcome which is neither $W$ nor a location which the deviating agents prefer. 

\begin{lemma}
	If all clients and facilities are located on a line, let the current state where $W$ wins have the location profile $(\vec{t};\vec{x}),$ and let $L_B$ be the set of clients in $L$ that strictly prefer $B$ over $W$, $L_m \subseteq L_B$, $ \vec{x'} = (B, B, \cdots, B)$ consist of $|L_m|$ occurrences of $B$. Let $A = GM^*(\vec{t};\vec{x}_{-L_m},\vec{x'})$. Then, we must have either $A = W$ or for all $i \in L_m$, $d(i, A) < d(i, W)$.  
	\label{lemma:lineshift}
\end{lemma}
\begin{proof}
    We note that we must have $B \neq W$ by the definition of $L_B$, assuming $L_B \neq \emptyset$. We first assume $B$ is to the right of $W$. To simplify notation, we will use $B$, $W$, and $A$ to refer to the locations of the facilities, as well as the facilities themselves. We will proceed with three cases: (i) $A \in [W, B]$ (meaning that $A$ is between $W$ and $B$); (ii) $A \in (-\infty, W)$ (meaning that $A$ is to the left of $W$); (iii) $A \in (B, \infty)$. For case (i), we note that every client in $L_m$ strictly prefers $B$ over $W$, which means that they are in the interval $(\frac{W+B}{2}, \infty)$. Now, given that $A \in [W, B]$, for any $i \in L_m$, we must have $d(i, A) < d(i, W)$, as desired. 

    Next, we consider case (ii), $A \in (-\infty, W)$. We will show that this case is impossible by showing that $W$ would always outperform $A$. To do this, we first observe that the costs for the clients in $L\setminus L_m$ and $T$ do not change, so we will consider changes in costs for clients in $L_m$. First, we consider clients in $L_m$ reported in $\vec{x}$ on top of or to the right of $W$ (in the interval $[W, \infty)$), let this set of clients be $L_R$. We have that for $i \in L_R$,  $cost_A(\vec{t};\vec{x}_{-i},B) - cost_A(\vec{t};\vec{x}) = cost_W(\vec{t};\vec{x}_{-i},B) - cost_W(\vec{t};\vec{x})$. Then, consider clients in $L_m$ reported in $\vec{x}$ on top of or to the left of $A$ (in the interval $(-\infty, A]$), let this set of clients be $L_L$. We have that for $i \in L_L$,  $cost_A(\vec{t};\vec{x}_{-i},B) - cost_A(\vec{t};\vec{x}) = d(A,W)+d(W,B) - d(x_i, A) > d(W,B) - d(x_i, A) - d(W,A) = cost_W(\vec{t};\vec{x}_{-i},B) - cost_W(\vec{t};\vec{x})$. Lastly, consider clients in $L_m$ reported in $\vec{x}$ in between $A$ and $W$ (in the interval $(A, W)$), let this set of clients be $L_M$. We have that for $i \in L_M$, $cost_A(\vec{t};\vec{x}_{-i},B) - cost_A(\vec{t};\vec{x}) = d(A,W)+d(W,B) - d(x_i, A) = d(x_i, W) + d(W,B) > d(W,B) - d(W,x_i) = cost_W(\vec{t};\vec{x}_{-i},B) - cost_W(\vec{t};\vec{x})$. Now, recall that $GM^*(\vec{t};\vec{x}) = W$, this means that $cost_W(\vec{t};\vec{x}) \leq cost_A(\vec{t};\vec{x})$, combining the previous results, we have that the cost of choosing $A$ does not become lower as compared to the cost of choosing $W$ for $(\vec{t};\vec{x}_{-L_m},\vec{x'})$, so $W$ would still outperform $A$ (note that in the case when they produce the same cost, since $GM^*(\vec{t};\vec{x}) = W$, $W$ is lower-numbered than $A$, meaning that $A$ would not be chosen after the deviation as well). 

    Then, we consider case (iii), $A \in (B, \infty)$. We will show that this case is also impossible. We first observe that comparing to $Med_L$ and $Med_R$ for $(\vec{t};\vec{x}),$ $Med_L'$ and $Med_R'$ for $(\vec{t};\vec{x}_{-L_m},\vec{x'})$ would either stay the same as $Med_L$ and $Med_R$ or move towards $B$. We note since $GM^*(\vec{t};\vec{x}) = W$, we must have $Med_L \in (-\infty,B)$, otherwise $W$ wouldn't have been chosen by Lemma \ref{lemma:optline}, since $B$ is a facility location and would be closer to $Med_L$ than $W$ otherwise. This means that $Med_L' \in (-\infty,B]$ by our previous observation. Since $A$ is selected by $GM^*$ after the deviation, Lemma  \ref{lemma:optline} also means that $Med_R' \in (B,\infty)$, since $A$ has to either be the closest facility to the right of $Med_R'$ or inside the interval $[Med_L',Med_R']$ to be selected. But this implies that $B$ is always inside the interval $[Med_L',Med_R']$. Thus $B$ would be chosen by $GM^*$ over $A$ according to Lemma \ref{lemma:optline}: either $A$ is not in $[Med_L',Med_R']$ so it is not optimal, or $A\in [Med_L',Med_R']$, in which case $B$ would still be chosen since it is to the left of $A$ and thus is lower-numbered than $A$. Therefore, $A$ would not be chosen by $GM^*$, which is a contradiction. 

    Finally, consider when $B$ is to the left of $W$. This case is mostly symmetric to the case where $B$ is to the right of $W$ (thus we can apply the same proof), since the mechanism $GM^*$ is completely symmetric except for the tie-breaking rule. There was only one part above in which we applied the fact that the tie-breaking rule favors facilities which are more to the left; this was in case (iii). Thus, we only need to argue that the result holds for the symmetric case of $A \in (-\infty, B)$. Similarly to case (iii) when $B$ is to the right of $W$, we first observe that comparing to $Med_L$ and $Med_R$ for $(\vec{t};\vec{x}),$ $Med_L'$ and $Med_R'$ for $(\vec{t};\vec{x}_{-L_m},\vec{x'})$ would either stay the same as $Med_L$ and $Med_R$ or move towards $B$. Note that since $GM^*(\vec{t};\vec{x}) = W$, we must have $Med_L \in (B,\infty)$, otherwise $W$ wouldn't have been chosen by Lemma \ref{lemma:optline} since $B$ is a facility location and is lower numbered than $W$ (since it is to the left of $W$). This means that $Med_L' \in [B, \infty)$ by our previous observation. This also means that no matter where $Med_R'$ is located, it holds that either $B\in [Med_L', Med_R']$, or $B$ is closer to $Med_L'$ than $A$. Since we have $A \notin [Med_L', Med_R']$,  by Lemma \ref{lemma:optline}, $A$ would not be chosen by $GM^*$ over $B$, which is a contradiction.

    Now, combining all the results above, we can conclude that either $A = W$ or for all $i \in L_m$, $d(i, A) < d(i, W)$. 
\end{proof}

Then, we can obtain the following results. 
\begin{theorem}
If all clients and facilities are located on a line, the strong price of anarchy and the strong price of stability of $GM^*$ are upper bounded by $\frac{n+k}{n-k}$.
	\label{thm:sneq}
\end{theorem}
% \begin{proof}[Proof Sketch.] We first note that if $O = X$ or is an SNE, it is done. Therefore, assume choosing $O$ under $GM^*$ is not an SNE. Since choosing $X$ under $GM^*$ is an SNE, we know that if all clients in $L_O$ alter their reported locations to be at $O$ in the SNE where $X$ wins, $O$ would still not outperform $X$. More formally, let the SNE where $X$ wins have the location profile $(\vec{t};\vec{x}),$ and set $ \vec{x'} = (O, O, \cdots, O)$. Then, since $X$ is a SNE, we must have $cost_X(\vec{t};\vec{x}_{-L_O},\vec{x'}) \leq cost_O(\vec{t};\vec{x}_{-L_O},\vec{x'})$. Otherwise, $GM^*(\vec{t};\vec{x}_{-L_O},\vec{x'}) \neq X$, and $X$ would not be a SNE. Then, similar to the proof of Theorem \ref{thm:poa}, with the fact that $i \in L_X$ prefers $X$ over $O$, which means that $\sum_{i \in L_X} d(i,X) \leq \sum_{i \in L_X} d(i,O)$ and triangle inequality, we can show that $\sum_{i \in T} d(i,X) + \sum_{i \in L} d(i,X) \leq \frac{n+(k-p)}{n}\sum_{i \in \C} d(i, O) + \frac{k-p}{n} \sum_{i \in \C}d(i,X)$. Rearranging this inequality and we can obtain $\sum_{i \in \C}d(i,X) \leq \frac{n+k}{n-k} \sum_{i \in \C}d(i,O)$ as desired.
% \end{proof}

\begin{proof}
	We first note that if $O = X$, then we are done, so assume $O\neq X$. 
    %Therefore, assume choosing $O$ under $GM^*$ is not an SNE.
    Let the SNE where $X$ wins have the location profile $(\vec{t};\vec{x}),$ and let $ \vec{x'} = (O, O, \cdots, O)$ consist of $p$ occurrences of $O$. 
    Then, since $X$ is an SNE, let $A = GM^*(\vec{t};\vec{x}_{-L_O},\vec{x'})$, it must be either $A = X$, or there exist some $j \in L_O$ such that $d(j, A) > d(j, X)$, otherwise all clients in $L_O$ could deviate to $O$ and improve their cost respectively. Since all clients in $L_O$ strictly prefer $O$ over $X$ by definition, by Lemma \ref{lemma:lineshift}, we must have either $A = W$ or for all $i \in L_O$, $d(i, A) < d(i, X)$. The latter case is a contradiction to $X$ being an SNE, so it must be that $W = A = GM^*(\vec{t};\vec{x}_{-L_O},\vec{x'})$, meaning that $cost_X(\vec{t};\vec{x}_{-L_O},\vec{x'}) \leq cost_O(\vec{t};\vec{x}_{-L_O},\vec{x'})$.
This means that  
	\[
	\begin{aligned}
		\sum_{i \in T}d(i,X) + \sum_{i \in L_O}d(O,X) + \sum_{i \in L_X}d(x_i,X) &\leq \sum_{i \in T}d(i,O) + \sum_{i \in L_O}d(O,O) + \sum_{i \in L_X}d(x_i,O)\\
		\sum_{i \in T}d(i,X)   &\leq \sum_{i \in T}d(i,O)  + \sum_{i \in L_X}d(x_i,O) \\
        & \quad- \sum_{i \in L_X}d(x_i,X) - \sum_{i \in L_O}d(O,X)
	\end{aligned}
	\]
	Now, by triangle inequality, we have that
	\[
	\begin{aligned}
		\sum_{i \in T}d(i,X) &\leq \sum_{i \in T}d(i,O)  + \sum_{i \in L_X}d(x_i,X) +  \sum_{i \in L_X}d(X,O) - \sum_{i \in L_X}d(x_i,X) - \sum_{i \in L_O}d(O,X)\\
		&= \sum_{i \in T}d(i,O)  + \sum_{i \in L_X}d(X,O)  - \sum_{i \in L_O} d(O,X)\\
		   &= \sum_{i \in T}d(i,O)  + (k-p)\cdot d(X,O)  - p\cdot d(O,X)\\
		   &= \sum_{i \in T}d(i,O)  + (k-2p)\cdot d(X,O) 
	\end{aligned}
	\]
	Then, we can see that
	\[
	\begin{aligned}
		\sum_{i \in T} d(i,X) + \sum_{i \in L} d(i,X) &\leq \sum_{i \in T}d(i,O)  + (k-2p)\cdot d(X,O) + \sum_{i \in L_O} d(i,X) + \sum_{i \in L_X} d(i,X)
	\end{aligned}
	\]
	Recall that $i \in L_X$ prefer $X$ over $O$, this means that $\sum_{i \in L_X} d(i,X) \leq \sum_{i \in L_X} d(i,O)$. Now, combined with triangle inequality, we can see that 
	\[
	\begin{aligned}
		\sum_{i \in T} d(i,X) + \sum_{i \in L} d(i,X) &\leq \sum_{i \in T}d(i,O)  + (k-2p)\cdot d(X,O) + \sum_{i \in L_O} d(i,O) \\
        & \quad + \sum_{i \in L_O} d(X,O) + \sum_{i \in L_X} d(i,O)\\
		&= \sum_{i \in \C}d(i,O)  + (k-2p)\cdot d(X,O) + p\cdot d(X,O)\\
		&= \sum_{i \in \C}d(i,O)  + (k-p)\cdot d(X,O)\\
		&\leq \sum_{i \in \C} d(i, O) + \frac{k-p}{n} \sum_{i \in \C}d(i,X) + \frac{k-p}{n} \sum_{i \in \C}d(i,O)\\
 		&= \frac{n+(k-p)}{n}\sum_{i \in \C} d(i, O) + \frac{k-p}{n} \sum_{i \in \C}d(i,X) 
	\end{aligned}
	\]
	The last inequality holds since $d(X,O) \leq d(i, X) + d(i,O)$ for all $i \in \C$. Now, rearrange the above inequality, and given that $p \geq 0$, we have 
	\[
	\begin{aligned}
		\frac{n-(k-p)}{n} \sum_{i \in \C}d(i,X) &\leq \frac{n+(k - p)}{n}\sum_{i \in \C}d(i,O)\\
		\sum_{i \in \C}d(i,X) &\leq \frac{n+(k-p)}{n-(k-p)} \sum_{i \in \C}d(i,O) \\
		\sum_{i \in \C}d(i,X) &\leq \frac{n+k}{n-k} \sum_{i \in \C}d(i,O).
	\end{aligned}
	\]
	Since the choice of $X$ is arbitrary, we can conclude that both the best and worst SNE would be within a factor of $\frac{n+k}{n-k}$ from the optimal solution, as desired. 
\end{proof}

However, we note that we can also obtain another bound.
\begin{theorem}
If all clients and facilities are located on a line, the strong price of anarchy and the strong price of stability of $GM^*$ are upper bounded by 3.
	\label{thm:sneq3}
\end{theorem}
\begin{proof}
	Recall that we have already shown in the proof for Theorem \ref{thm:sneq}
	$$\sum_{i \in \C} d(i,X) \leq \sum_{i \in \C}d(i,O)  + (k-p)\cdot d(X,O)$$
	By triangle inequality and that $i \in L_X$ prefers $X$ over $O$, we have that 
	\[
	\begin{aligned}
		\sum_{i \in \C} d(i,X) &\leq \sum_{i \in \C}d(i,O)  + \sum_{i \in L_X} d(i,O) + \sum_{i \in L_X} d(i,X)\\
		&\leq \sum_{i \in \C}d(i,O)  + \sum_{i \in L_X} d(i,O) + \sum_{i \in L_X} d(i,O)\\
		&\leq \sum_{i \in \C}d(i,O) + \sum_{i \in \C}d(i,O) + \sum_{i \in \C}d(i,O)\\
		& = 3\sum_{i \in \C}d(i,O)
	\end{aligned}
	\]
	as desired.
\end{proof}

Now, combining Theorem \ref{thm:sneq} and Theorem \ref{thm:sneq3}, we can obtain the following result. 
\begin{theorem}
If all clients and facilities are located on a line, the strong price of anarchy and the strong price of stability of $GM^*$ are at most 
\[
\left\{\begin{array}{cc}
\frac{n+k}{n-k} &\text{if  }  k \leq \frac{n}{2} \\
3 &\text{if  } k>\frac{n}{2}
\end{array}\right.
\]
and these bounds are asymptotically tight.
	\label{thm:sneqfull}
\end{theorem}
\begin{proof}

We note that by Theorem \ref{thm:sneq} and Theorem \ref{thm:sneq3}, we can conclude that for any $k$, the strong price of anarchy and the strong price of stability are upper bounded by $\min\left(3, \frac{n+k}{n-k} \right)$. Since $\frac{n+k}{n-k} \leq 3$ when $k \leq \frac{n}{2}$ and $\frac{n+k}{n-k} > 3$ when $k > \frac{n}{2}$, we have the upper bound result of the theorem. 
%Let $c = \frac{n}{k}, c \geq 1$, $f(c) = \frac{n+k}{n-k}=\frac{n+n/c}{n-n/c} = \frac{c+1}{c-1}$. We note that $f'(c) = -\frac{2}{(c-1)^2}$, so $f'(c) < 0$ when $c > 1$. This means that when $k < n$, $f(n)$ is monotone decreasing, and $f(c) = 3$ when $c = 2, k = \frac{n}{2}$. This means that when $k > \frac{n}{2}$, $f(c) > 3$, therefore we can conclude that 
%The strong price of anarchy and the strong price of stability are at most 
% \[
% \left\{\begin{array}{cc}
% \frac{n+k}{n-k} & , k \leq \frac{n}{2} \\
% 3 &,  k>\frac{n}{2}
% \end{array}\right.
% \]

We will now show that the above bound is tight. First, consider the case when $k > \frac{n}{2}$. Consider an instance as shown in Figure \ref{fig:3bound}. There are three locations $O, A$ and $N$ on a line, with facility locations $O, N$ and ties resolved in favor of $N$. There are $n/2+1$ clients in $L$ at location $A$ and the rest of the clients in $\C$ at $O$. We also have $d(O,A) = 1 + \epsilon, d(A,N) = 1$, with $\epsilon > 0$ sufficiently small. We note that $O$ is the optimal solution with cost $\frac{n}{2}+ 1 + \left( \frac{n}{2}+ 1 \right)\epsilon$. However, the only SNE has $GM^*$ choosing $N$. Consider the case where all the clients at $A$ report to be at $N$, in which case $GM^*(\vec{t};\vec{x})=N$ with cost $\frac{3n}{2} + \left(\frac{n}{2}-1\right)\epsilon$, with $n \rightarrow \infty, \epsilon \rightarrow 0$, we have that the strong price of anarchy and the strong price of stability are 3 as desired. 

Now, consider the case $k \leq \frac{n}{2}$. Consider an instance as shown in Figure \ref{fig:bound}. There are four locations $O, A, B$ and $C$ on a line, with possible facility locations $O, A$. All $k$ clients in $L$ are at location $C$ and all clients in $T$ are at $B$. We also have $d(O,B) = \frac{n-2k}{2(n-k)}, d(B,C) = \frac{k}{2(n-k)}+\epsilon, d(O,C) = \frac{1}{2} + \epsilon, d(C,A) = \frac{1}{2}-\epsilon$, with $\epsilon > 0$ sufficiently small. We note that $O$ is the optimal solution with cost $\frac{n-k}{2} + \epsilon k$. However, the only SNE has $GM^*$ choosing $A$. Consider the case where all the clients at $C$ report to be at $A$, in which case $GM^*(\vec{t};\vec{x})=A$ with cost $\frac{n+k}{2} - \epsilon k$. With $\epsilon \rightarrow 0$, we have that the strong price of anarchy and the strong price of stability are $\frac{n+k}{n-k}$, as desired. 
\end{proof}

	\begin{figure}[h!]
\centering
\begin{tikzpicture}[every node/.style={scale=1.1}, node distance=2cm]

    % Nodes
\node[draw, rectangle, fill=cyan!30] (O) at (0,0) {O};
\node[draw, circle, fill=cyan!30] (A) at (3,0) {A};
\node[draw, rectangle, fill=green!50] (N) at (6,0) {N};

% Edges
\draw[-] (O) -- (A) node[midway, above] {$1 + \varepsilon$};
\draw[-] (A) -- (N) node[midway, above] {$1$};

% Bottom labels
\node at (0,-1) {$\frac{n}{2} - 1$};
\node at (3,-1) {$\frac{n}{2} + 1$};

\end{tikzpicture}
\caption{A lower bound instance for SPoS and SPoA.
%of a facility location problem. %There are two facility locations $O,N$, with $n/2+1$ clients in $L$ at $A$ and the rest of the clients in $\C$ at $O$ (true locations), with the distance between them labeled.
}
\label{fig:3bound}
\end{figure}

\begin{figure}[h!]
\centering
\begin{tikzpicture}[every node/.style={scale=1.1}, node distance=2cm]

% Nodes
\node[draw, rectangle, fill=green!50] (O) at (0,0) {O};
\node[draw, circle, fill=cyan!30] (B) at (2.5,0) {B};
\node[draw, circle, fill=cyan!30] (C) at (5,0) {C};
\node[draw, rectangle, fill=green!] (A) at (10,0) {A};

% Edges
\draw[-] (O) -- (B);
\draw[-] (B) -- (C) node[midway, above=0.4] {$\dfrac{k}{2(n - k)} +\epsilon$};
\draw[-] (C) -- (A) node[midway, above] {$\frac{1}{2} - \epsilon$};

% Top label
\node at (1.5, 1) {$\dfrac{n - 2k}{2(n - k)}$};

% Bottom labels
\node at (3,-1) {$n - k$};
\node at (5,-1) {$k$};
\end{tikzpicture}
\caption{A lower bound instance for SPoS and SPoA.}
%An instance of a facility location problem. There are two facility locations $O,A$, with all $n-k$ clients in $T$ at $B$, all $k$ clients in $L$ at $C$ (true locations), with the distance between them labeled.}
\label{fig:bound}
\end{figure}

\section{Conclusion}

In this work, we analyzed the properties of equilibrium solutions for strategic facility location. We showed that as long as the fraction of strategic players is not too large, the cost of these solutions can be much better than those created by strategyproof mechanisms, especially for general metric spaces. 

Some clear open questions remain. One is about convergence: while we proved the existence of various equilibrium solutions, it would be nice to analyze which processes and dynamics converge to such solutions for this setting. Another direction involves extending our results to more general equilibrium concepts. In fact, it should be possible to generalize some of our results to $\ell$-strong equilibrium, which is a solution from which no coalition of size $\ell$ can benefit by deviating \cite{epstein2007strong}. Such solutions always exist for our setting as long as $\ell\leq\frac{k}{2}$, and it should be possible to extend our price of anarchy and stability bounds to such solutions. %For example, the price of anarchy for $\ell$-strong equilibrium on a line metric is bounded by $\frac{n+2(k-\ell)}{n-2(k-\ell)}$. 
Quantifying the cost of other notions of stability, such as when all clients are strategic, but can only report a location close to their true location, remains the subject of future work.

%\subsection*{Acknowledgments} We are grateful to Wennan Zhu for formulating the facility location and voting setting with a limited number of liars. 

%%
%% The next two lines define the bibliography style to be used, and
%% the bibliography file.
% \bibliographystyle{ACM-Reference-Format}
\bibliography{refs.bib}

@inproceedings{gravin2025approximation,
  title={Approximation Guarantees of Median Mechanism in $\mathbb{R}^d$},
  author={Gravin, Nikolai and Jia, Jianhao},
  booktitle={Proceedings of the 57th Annual ACM Symposium on Theory of Computing},
  pages={495--506},
  year={2025}
}

@inproceedings{aloupis2010improved,
  author    = {Greg Aloupis and Mordecai Ben-Or and Shiri Chechik and Michal Feldman and Ariel D. Procaccia and Moshe Tennenholtz},
  title     = {Improved Approximation Ratio for Strategyproof Facility Location on a Cycle},
  booktitle = {Proceedings of the 6th International Workshop on Internet and Network Economics (WINE)},
  series    = {Lecture Notes in Computer Science},
  volume    = {6484},
  pages     = {379--390},
  year      = {2010},
  publisher = {Springer},
  doi       = {10.1007/978-3-642-17572-5_30}
}

@article{pathak2008leveling,
  title={Leveling the playing field: Sincere and sophisticated players in the Boston mechanism},
  author={Pathak, Parag A and S{\"o}nmez, Tayfun},
  journal={American Economic Review},
  volume={98},
  number={4},
  pages={1636--1652},
  year={2008},
  publisher={American Economic Association},
DOI = {10.1257/aer.98.4.1636}
}

@article{groseclose2010sincere,
  title={Sincere versus sophisticated voting in Congress: Theory and evidence},
  author={Groseclose, Tim and Milyo, Jeffrey},
  journal={The Journal of Politics},
  volume={72},
  number={1},
  pages={60--73},
  year={2010},
  publisher={Cambridge University Press New York, USA},
doi={10.1017/s0022381609990478}
}

@incollection{lebon2018sincere,
  TITLE = {{Sincere voting, strategic voting : A laboratory experiment using alternative proportional systems}},
  AUTHOR = {Baujard, Antoinette and Igersheim, Herrade and Gavrel, Fr{\'e}d{\'e}ric and Laslier, Jean-Fran{\c c}ois and Lebon, Isabelle},
  URL = {https://shs.hal.science/halshs-01652699},
  BOOKTITLE = {{The Many Faces of Strategic Voting}},
  EDITOR = {John Aldrich and Andr{\'e} Blais and Laura B. Stephenson},
  PUBLISHER = {{The University of Michigan Press}},
  Chapter = {10},
  PAGES = {203-231},
  YEAR = {2018},
  DOI = {10.3998/mpub.9946117},
  HAL_ID = {halshs-01652699},
  HAL_VERSION = {v1},
}

@inproceedings{walsh2021strategy,
  title={Strategy proof mechanisms for facility location at limited locations},
  author={Walsh, Toby},
  booktitle={PRICAI 2021: Trends in Artificial Intelligence: 18th Pacific Rim International Conference on Artificial Intelligence, PRICAI 2021, Hanoi, Vietnam, November 8--12, 2021, Proceedings, Part I 18},
  pages={113--124},
  year={2021},
  publisher={Springer},
address="Cham",
doi={10.1007/978-3-030-89188-6_9}
}

@inproceedings{lam2024proportional,
  author = {Lam, Alexander and Aziz, Haris and Li, Bo and Ramezani, Fahimeh and Walsh, Toby},
title = {Proportional Fairness in Obnoxious Facility Location},
year = {2024},
isbn = {9798400704864},
publisher = {International Foundation for Autonomous Agents and Multiagent Systems},
address = {Richland, SC},
booktitle = {Proceedings of the 23rd International Conference on Autonomous Agents and Multiagent Systems},
pages = {1075–1083},
numpages = {9},
location = {Auckland, New Zealand},
series = {AAMAS '24},
doi={10.5555/3635637.3662963}
}

@inproceedings{chan2021mechanism,
    title     = {Mechanism Design for Facility Location Problems: A Survey},
  author    = {Chan, Hau and Filos-Ratsikas, Aris and Li, Bo and Li, Minming and Wang, Chenhao},
  booktitle = {Proceedings of the Thirtieth International Joint Conference on
               Artificial Intelligence, {IJCAI-21}},
  publisher = {International Joint Conferences on Artificial Intelligence Organization},
  editor    = {Zhi-Hua Zhou},
  pages     = {4356--4365},
  year      = {2021},
  month     = {8},
  doi       = {10.24963/ijcai.2021/596}
}

@article{aziz2022strategyproof,
  title = {Proportionality-based fairness and strategyproofness in the facility location problem},
journal = {Journal of Mathematical Economics},
volume = {119},
pages = {103129},
year = {2025},
issn = {0304-4068},
doi = {https://doi.org/10.1016/j.jmateco.2025.103129},
author = {Haris Aziz and Alexander Lam and Barton E. Lee and Toby Walsh}
}

@inproceedings{feldman2016voting,
  author = {Feldman, Michal and Fiat, Amos and Golomb, Iddan},
title = {On Voting and Facility Location},
year = {2016},
isbn = {9781450339360},
publisher = {Association for Computing Machinery},
address = {New York, NY, USA},
url = {https://doi.org/10.1145/2940716.2940725},
doi = {10.1145/2940716.2940725},
booktitle = {Proceedings of the 2016 ACM Conference on Economics and Computation},
pages = {269–286},
numpages = {18},
location = {Maastricht, The Netherlands},
series = {EC '16}
}

@article{anshelevich2017randomized,
  title={Randomized social choice functions under metric preferences},
  author={Anshelevich, Elliot and Postl, John},
  journal={Journal of Artificial Intelligence Research},
  volume={58},
  pages={797--827},
  year={2017},
  doi={10.5555/3176764.3176784},
  publisher = {AI Access Foundation},
address = {El Segundo, CA, USA},
number = {1},
}

@article{procaccia2013approximate,
  title={Approximate mechanism design without money},
  author={Procaccia, Ariel D and Tennenholtz, Moshe},
  journal={ACM Transactions on Economics and Computation (TEAC)},
  volume={1},
  number={4},
  pages={1--26},
  year={2013},
  publisher = {Association for Computing Machinery},
address = {New York, NY, USA},
volume = {1},
number = {4},
issn = {2167-8375},
doi = {10.1145/2542174.2542175}
}

@article{alon2010strategyproof,
  title={Strategyproof approximation of the minimax on networks},
  author={Alon, Noga and Feldman, Michal and Procaccia, Ariel D and Tennenholtz, Moshe},
  journal={Mathematics of Operations Research},
  volume={35},
  number={3},
  pages={513--526},
  year={2010},
  publisher={INFORMS},
  doi = {10.1287/moor.1100.0457}
}

@article{tang2023strategyproof,
  title={Strategyproof facility location with limited locations},
  author={Tang, Zhong-Zheng and Wang, Chen-Hao and Zhang, Meng-Qi and Zhao, Ying-Chao},
  journal={Journal of the Operations Research Society of China},
  volume={11},
  number={3},
  pages={553--567},
  year={2023},
  publisher={Springer},
doi={10.1007/s40305-021-00378-1}
}

@article{balkanski2024randomized,
  title={Randomized strategic facility location with predictions},
  author={Balkanski, Eric and Gkatzelis, Vasilis and Shahkarami, Golnoosh},
  journal={Advances in Neural Information Processing Systems (NeurIPS)},
  volume={37},
  pages={35639--35664},
  year={2024}
}

@article{filos2024distortion,
  title = {The distortion of distributed facility location},
journal = {Artificial Intelligence},
volume = {328},
pages = {104066},
year = {2024},
issn = {0004-3702},
doi = {https://doi.org/10.1016/j.artint.2024.104066},
author = {Aris Filos-Ratsikas and Panagiotis Kanellopoulos and Alexandros A. Voudouris and Rongsen Zhang}
}

@article{caragiannis2016truthful,
  title={Truthful facility assignment with resource augmentation: An exact analysis of serial dictatorship},
  author={Caragiannis, Ioannis and Filos-Ratsikas, Aris and Frederiksen, S{\o}ren Kristoffer Stiil and Hansen, Kristoffer Arnsfelt and Tan, Zihan},
  journal= {Mathematical Programming},
  year = {2024},
 volume = {203},
number = {1},
pages = {901--930},
doi={10.1007/s10107-022-01902-8}
}

@article{epstein2007strong,
  title = {Strong equilibrium in cost sharing connection games},
journal = {Games and Economic Behavior},
volume = {67},
number = {1},
pages = {51-68},
year = {2009},
issn = {0899-8256},
doi = {https://doi.org/10.1016/j.geb.2008.07.002},
author = {Amir Epstein and Michal Feldman and Yishay Mansour}
}

@article{andelman2009strong,
  title={Strong price of anarchy},
  author={Andelman, Nir and Feldman, Michal and Mansour, Yishay},
  journal={Games and Economic Behavior},
  volume={65},
  number={2},
  pages={289--317},
  year={2009},
  publisher={Elsevier},
  doi = {https://doi.org/10.1016/j.geb.2008.03.005}
}

@article{alon2009strategyproof,
  title={Strategyproof approximation mechanisms for location on networks},
  author={Alon, Noga and Feldman, Michal and Procaccia, Ariel D and Tennenholtz, Moshe},
  journal={arXiv preprint arXiv:0907.2049},
  year={2009}
}

@article{schummer2002strategy,
  title={Strategy-proof location on a network},
  author={Schummer, James and Vohra, Rakesh V},
  journal={Journal of Economic Theory},
  volume={104},
  number={2},
  pages={405--428},
  year={2002},
  publisher={Elsevier},
doi = {https://doi.org/10.1006/jeth.2001.2807}
}

\end{document}